%% file: ms.tex
\documentclass[12pt]{article}
\usepackage{collab,razb, bussproofs}
\usetikzlibrary{decorations.pathreplacing}

\input{razb}

\begin{document}
\title{On CDCL-based proof systems with the ordered decision strategy}
\author{Nathan Mull\thanks{University of Chicago, Department of Computer Science, {\tt nmull@cs.uchicago.edu.}} \and Shuo Pang\thanks{University of Chicago, Department of Mathematics, {\tt spang@math.uchicago.edu}.} \and Alexander Razborov\thanks{University of Chicago, USA, {\tt razborov@math.uchicago.edu} and Steklov Mathematical Institute, Moscow, Russia, {\tt razborov@mi-ras.ru}.}}
\maketitle
\abstract{
    \input{abstract}
}
\section{Introduction}
    \input{introduction}
\section{Preliminaries and main results} \label{sec:prel}
    \input{prel}
\section{$\cdcl(\pid, \decisionl)=_p\pi$-Ordered} \label{sec:last-l}
    \input{last_l}
\section{$\cdcl(\pid, \firstl)=_p$ Resolution} \label{sec:first-l}
    \input{first_l}
\section{Width lower bound} \label{sec:width}
    \input{width}
\section{Conclusion} \label{sec:concl}
    \input{conclusion}
\bibliographystyle{plain}
\bibliography{refs}
\end{document}

%% file: razb.tex
\newlength{\longeqmarginwidth}
\newlength{\longeqwidth}
\newlength{\longeqskiplength}
\newcommand{\Vzero}[2]{\mathord{\stackrel{\textstyle\kern1pt\circ}
{\smash V\vbox to6pt{}}}\vphantom{V}_{#1}^{#2}}

\newcommand{\function}[2]{:#1 \longrightarrow #2}
\newcommand{\of}[1]{\left( #1 \right)}

\newcommand{\set}[2]{\left\{\hspace{0.2ex} #1 \left|\: #2
\right. \right\}}

\newcommand{\df}{\stackrel{\rm def}{=}}

\newcommand{\res}{\mbox{Res}}

\newcommand{\unaryinf}[2]{{\strut {\displaystyle #1}\over #2}}
\newcommand{\binaryinf}[3]{\frac{\strut{\displaystyle  #1}\hspace{1cm}
{\displaystyle #2}}{#3}}
\newcommand{\ternaryinf}[4]{\frac{\strut{\displaystyle  #1}\hspace{1cm}
\strut{\displaystyle #2}\hspace{1cm} {\displaystyle #3}}{#4}}

\newcounter{operator}

%% file: abstract.tex
We prove that conflict-driven clause learning SAT-solvers with the ordered decision strategy and the DECISION
learning scheme are equivalent to ordered resolution. We also prove that, by
replacing this learning scheme with its opposite that stops after the first
new clause when backtracking, it becomes equivalent to general resolution. To
the best of our knowledge, this is the first theoretical study of the
interplay between specific decision strategies and clause learning.

For both results, we allow nondeterminism in the solver's ability to perform unit propagation, conflict analysis, and restarts, in a way that is similar to previous works in the literature.
To aid the presentation of our results, and possibly future research, we define a model and language for discussing CDCL-based proof systems that allows for succinct and precise theorem statements.

%% file: introduction.tex
%
%

SAT-solvers have become standard tools in many application domains such as hardware verification, software verification, automated theorem proving, scheduling and computational biology (see \cite{2006-Gupta-Ganai-Wang, 2008-Ivancic-et-al, 2009-Bonet-StJohn,2011-McGregor, 2017-Cohen-et-al} among the others).
Since their conception in the early 1960s, SAT-solvers have become significantly more efficient, but they have also become significantly more complex.
Consequently, there has been increasing interest in understanding the theoretical limitations and strengths of contemporary SAT-solvers.
Much of the recent literature has focused on the connections between SAT-solvers and subsystems of the resolution proof system originally introduced in \cite{1938-Blake, 1965-Robinson}.

This connection essentially started with the Davis-\-Putnam-\-Logemann-\-Loveland procedure (DPLL) \cite{1962-Davis-Logemann-Loveland}, a backtracking search algorithm that builds partial assignments one literal at a time until a satisfying assignment is found or all assignments have been exhausted.
Since DPLL is sound and complete, its computational trace when applied to an unsatisfiable formula is a \textit{proof} of unsatisfiability.
It is generally accepted as a folklore result that the computational trace of DPLL on an unsatisfiable formula can be converted into a tree-like resolution refutation.
Thus, tree-like resolution lower bounds imply DPLL running time lower bounds.
And in some sense, these lower bounds are tight: DPLL, given oracle access to a tree-like resolution refutation $\Pi$ of the input formula, can run in time that is polynomial in the length of $\Pi$.
That is, DPLL is essentially equivalent to tree-like resolution and thus can be viewed as a propositional proof system in the Cook-Reckhow sense \cite{1979-Cook-Reckhow}.

Nearly all contemporary SAT-solvers are variants of DPLL augmented with modern algorithmic techniques and heuristics.
The technique most often credited for their success is \textit{conflict-driven clause learning} (CDCL)
\cite{1997-Bayardo-Schrag, 1999-Grasp}, so these solvers are interchangeably called CDCL SAT-solvers, CDCL solvers, or simply CDCL (for further information regarding the design of SAT-solvers, see the 2009 \textit{Handbook of Satisfiability} \cite{2009-SAT-Handbook}).
Just as with DPLL, the computational trace of CDCL can be converted into a resolution refutation, but may no longer be tree-like or even regular.
Thus, general resolution lower bounds imply CDCL running time lower bounds, but it is unclear \textit{a priori} whether these bounds are tight in the same sense as above.

The line of work on the question of whether CDCL solvers simulate general resolution was initiated by Beame et al.\ \cite{2004-Beame-Kautz-Sabharwal} and continued by many others \cite{2005-Van-Gelder, 2006-Nieuwenhuis-Oliveras-Tinelli, 2008-Hertel-Bacchus-Pitassi-Van-Gelder, 2008-Buss-Hoffmann-Johannsen, 2010-Ben-Sasson-Johannsen, 2011-Pipatsrisawat-Darwiche, 2011-Atserias-Fichte-Thurley, 2016-Elffers-et-al}.
The primary difference between all these papers is in the details of the model, the models considered by Pipatsrisawat and
Darwich \cite{2011-Pipatsrisawat-Darwiche} and Atserias et al.\ \cite{2011-Atserias-Fichte-Thurley} being perhaps the most faithful to actual implementations of CDCL SAT-solvers.
But almost all models appearing in the literature make a few nonstandard assumptions.
\begin{enumerate}
    \item
        \textit{Very frequent restarts.}
        The solver restarts roughly $O(n^2)$ times for every clause in the given resolution refutation $\Pi$.
        Though many solvers do restart frequently in practice \cite{2015-Biere-Frohlich}, it is unclear if this is
        really necessary for the strength of CDCL.
    \item
        \textit{No clause deletion policy.}
        The solver has to keep every learned clause.
        In practice, some solvers periodically remove half of all learned clauses \cite{2009-Audemard-Simon}.
    \item
        \textit{Nondeterministic decision strategy.}
        The solver uses oracle access to $\Pi$ to construct a very particular decision strategy.
        In practice, solvers use heuristics \cite{1999-Marques-Silva, 2001-Chaff, 2015-Liang-et-al}.
\end{enumerate}

It is natural to ask whether these assumptions can be weakened or removed entirely.
In this respect, the first two assumptions have become topics of recent interest.
With regards to the first, much research has been dedicated to the study of \textit{nonrestarting} SAT-solvers \cite{2005-Van-Gelder, 2008-Buss-Hoffmann-Johannsen, 2014-Buss-Kolodziejczyk, 2014-Bonet-Buss-Johannsen, 2014-Beame-Sabharwal}.
The exact strength of CDCL without restarts is still unknown and, arguably, makes for the most interesting
open problem in the area.
With regards to the second, Elffers et.\ al.\ \cite{2016-Elffers-et-al} proved size-space tradeoffs in a very tight model of CDCL, which may be interpreted as results about aggressive clause deletion policies.

In this paper we are primarily concerned with the third assumption, i.e.\ how much does the efficiency of CDCL-solvers
depend on the nondeterminism in the decision strategy?
To the best of our knowledge, this has not yet been considered in the literature, perhaps due
to the complexity of decision strategies in contemporary SAT-solvers.
We study a simple decision strategy that we call the \textit{ordered} decision strategy and,
as a historical motivation, we remark that it is identical to the strategy studied by Beame et.\ al.\ \cite{2002-Beame-Karp-Pitassi-Saks} in the context of DPLL without clause learning.
It is defined naturally: when the solver has to choose a variable to assign, the ordered decision strategy dictates that it chooses the smallest unassigned variable according to some fixed order.
There is still a choice in whether to fix the variable to $0$ (\textit{false}) or $1$ (\textit{true}), and
we allow the solver to make this choice nondeterministically. If unit propagation is used, the solver may assign variables out
of order; a unit clause does not necessarily correspond to  the smallest unassigned variable.
This possibility to ``cut the line'' is precisely what makes the situation much more subtle and nontrivial.

Thus, our motivating question is the following:
\begin{displayquote}
    \textit{
        Is there a family of contradictory CNFs $\{\tau_n\}_{n = 1}^\infty$ that possess polynomial
        size resolution refutations but require superpolynomial time for CDCL using the ordered decision scheme, for any order?
    }
\end{displayquote}
Before describing our contributions towards this question, let us briefly review analogous separations in
the context of proof and computational complexities.
Bonet et.\ al.\ \cite{2000-Bonet-Esteban-Galesi-Johannsen} proved that a certain family of formulas requires exponential-sized ordered resolution refutations but have polynomial-sized regular resolution refutations.
Bollig et.\ al.\ \cite{1999-Bollig-et-al} proved that a certain boolean function requires exponential-sized ordered binary decision diagrams (OBDDs) but have polynomial-sized general BDDs.
These results tell us that order tends to be a strong restriction, and the above question asks whether this same phenomenon occurs for CDCL.
It is also worth noting that this question may be motivated as a way of understanding the strength of \textit{static} decision strategies such as MINCE \cite{2001-Aloul-et-al} and FORCE \cite{2003-Aloul-et-al}.
But since such decision strategies are rarely used in practice we will not dwell on this anymore.

\subsection*{Our contributions}

Per the discussion above, a proof system that captures any class of CDCL solvers should be no stronger than general resolution.
It can also be reasonably expected (and in two particular situations will be verified below as easy directions of Theorems \ref{thm:last-l}, \ref{thm:main}) that CDCL with the ordered decision strategy should be at least as strong as ordered resolution w.r.t.\ the same order.
Our main results show that, depending on the learning scheme employed, both of these extremes can be attained.
More specifically, we prove
\begin{enumerate}
    \item
        CDCL with the ordered decision strategy and a learning scheme we call $\decisionl$ is equivalent to ordered resolution (Theorem \ref{thm:last-l}).
        In particular, it does not simulate general resolution.

    \item CDCL with the ordered decision strategy and a learning scheme we call \textsf{FIRST-L} is equivalent to general resolution (Theorem \ref{thm:main}).
\end{enumerate}

\begin{remark}
As the name suggests, $\decisionl$ is the same as the so-called DECISION learning scheme in the literature.\footnote{We use this slightly different name so that it fits our naming conventions below.}
Hence these two results, taken together, go somewhat against the ``common wisdom".
Namely, it turns out that in the case of ordered decision strategy, an assertive learning scheme is badly out-performed by a scheme that, to the best of our knowledge, has not been used in practice.
That said, \textsf{FIRST-L} is similar to the learning scheme FirstNewCut \cite{2004-Beame-Kautz-Sabharwal}, and both schemes have the property that they are designed somewhat artificially to \textit{target} particular resolution steps in a given refutation.
\end{remark}

We also prove linear width lower bounds for the second variant (Theorem \ref{thm:width}), which are in sharp contrast with the size-width relationship for general resolution proved by Ben-Sasson and Wigderson \cite{1999-Ben-Sasson-Wigderson}.

In all these results, the CDCL solver may arbitrarily choose the conflict/unit clause if there are several, may elect not to do conflict analysis/unit propagations at all, and may restart at any time.
We work under these assumptions in part because our intention here is to focus on the impact of decision strategies.
But this substantial amount of nondeterminism also allows us to identify two proof systems that are, more or less straightforwardly, equivalent to the corresponding CDCL variant.
Determining the exact power of these systems constitutes the main technical part of this paper.

The first proof system might be of independent interest; we call it {\em half-ordered resolution}. For
a given order on the variables, ordered resolution can be alternately described by the requirement that
in every application of the resolution rule, the resolved variable is larger than any other variable
appearing in both of the two antecedent clauses. We relax this requirement by asking that this property holds only
for \textit{one} of them, which  reflects the inherent asymmetry in resolution rules resulting from clause learning
in CDCL solvers. Somewhat surprisingly (at least to us), it turns out (Theorem \ref{thm:p1}) that this relaxation
does not add any extra power, and half-ordered resolution is polynomially equivalent to ordered resolution w.r.t.\ the same order.

The second proof system (let us call it $P_0$) extends half-ordered resolution and is more auxiliary in nature. It is based
on the observation that with the amount of nondeterminism we allow, all trails\footnote{A \textit{trail} is essentially an ordered partial assignment constructed by CDCL during its execution.} that a CDCL solver manages to create can be
easily recreated when needed.
Accordingly, the system works with lines of two types, one for clauses and another for trails.
Clauses entail nontrivial trails via a unit propagation rule while trails can be used to enhance the half-ordered resolution rule.
We show that $P_0$ is polynomially equivalent to resolution (Theorem \ref{thm:p0}), and since it is by far our most
difficult result, let us reflect a bit on the ideas in its proof.

Like other CDCL-based proof systems, $P_0$ is not closed under restrictions or weakening, so many standard methods no longer apply.
Instead, we use two operations on resolution proofs (lifting and variable deletion) in tandem with some additional structural information to give us a fine-grained understanding of the size and structure of the general resolution refutation being simulated.
The properties of these operators allow for a surgery-like process; we simulate small local pieces of the refutation and then stitch them together into a new global refutation.

Theorem \ref{thm:main} may at first seem strange to those unfamiliar with CDCL-based proof systems.
We allow nondeterminism in parts of CDCL that are typically deterministic, but then prove an upper bound rather than a lower bound.
Morally speaking, we show that it is possible to \textit{redistribute} some of the nondeterminism in the decision strategy to different parts of CDCL while maintaining the ability to simulate general resolution.
Although the resulting algorithm deviates fairly drastically from standard implementations of CDCL, this deviation is not unprecedented.
The correspondence between proof systems and algorithms here is very similar to the correspondence between \textit{regWRTI} and a variant of CDCL with similar features called DLL-LEARN, both introduced by Buss et.\ al.\ \cite{2008-Buss-Hoffmann-Johannsen}.
Static proof systems are easier to analyze, but these nonstandard sources of nondeterminism manifest themselves naturally when translating CDCL into a static proof system.
Like our systems, lower bounds on \textit{regWRTI} imply lower bounds on standard implementations of CDCL, but recent results for \textit{regWRTI} have been upper bounds \cite{2014-Buss-Kolodziejczyk, 2014-Bonet-Buss-Johannsen}.
Of course, upper bounds are interesting in their own right, but even for those who are lower bound inclined, these results have value: they demonstrate, often nontrivially, what convenient features of simple proof systems must be dropped in order to prove separations.

Finally, in order to aid the above work (and, perhaps, even facilitate further research in the area), we present a model and language for studying CDCL-based proof systems.
This model is not meant to be novel, and is heavily influenced by previous work \cite{2006-Nieuwenhuis-Oliveras-Tinelli, 2011-Atserias-Fichte-Thurley, 2016-Elffers-et-al}.
However, the primary goal of our model is to \textit{highlight} possible nonstandard sources of nondeterminism in variants of CDCL, as opposed to creating a model
completely faithful to applications.
For example, Theorem \ref{thm:main} can be written in this language as:
\begin{displayquote} \em
    For any order $\pi$, \textnormal{\textsf{CDCL}(\textsf{FIRST-L}, \textsf{$\pi$-D})} is equivalent to general resolution.
\end{displayquote}
We will also try to pay a special attention to finer details of the model sometimes left implicit in previous works.
This entails several subtle choices to be made, and we interlace the mathematical description of our model with informal
discussion of these choices.

The paper is organized as follows. In Section \ref{sec:prel} we give all necessary definitions and formulate our main results as we go along.

In Section \ref{sec:last-l} we prove Theorem \ref{thm:last-l} on the power of $\cdcl$ with the ordered decision strategy and the $\decisionl$ learning strategy. Section \ref{sec:proofcomplexity} contains proof-complexity theoretic arguments
about half-ordered resolution, while in Section \ref{sec:equiv1} we establish its translation to the language of $\cdcl$.

In Section \ref{sec:first-l} we prove Theorem \ref{thm:main} on the power of $\cdcl$ with the ordered decision strategy and the $\firstl$ learning strategy. To that end, in Section \ref{sec:equiv2} we show the equivalence of this system to the proof system $P_0$
mentioned above, and in Section \ref{sec:main} we establish that $P_0$ is actually equivalent to general resolution
(Theorem \ref{thm:p0}).

In Section \ref{sec:width} we prove Theorem \ref{thm:width} that, roughly speaking, states that the simulation provided by Theorem \ref{thm:main} fails extremely badly with respect to width. Among other things, this implies that there does not seem to exist any useful width-size relation in the context of CDCL with ordered decision strategy.

We conclude in Section \ref{sec:concl} with a few remarks and suggestions for future work.

%% file: prel.tex
Throughout the paper, we assume that the set of propositional variables is fixed as $V\df\{x_1,\ldots,x_n\}$.
A {\em literal} is either a propositional variable or its negation.
We will sometimes use the abbreviation $x^0$ for $\bar x$ and $x^1$ for $x$ (so that the Boolean assignment $x=a$ {\em satisfies} the literal $x^a$).
A {\em clause} is a set of literals, thought of as their disjunction, in which no variable appears together with its negation.
For a clause $C$, $\var(C)$ is the set of variables appearing in $C$.
A {\em CNF} is a set of clauses thought of as their conjunction.
For a CNF $\tau$, $\var(\tau)$ is the set of variables appearing in $\tau$, i.e.\ the union of $\var(C)$ for all $C \in \tau$.
We denote the empty clause by 0.
The {\em width} of a clause is the number of literals in it. A {\em $w$-CNF} is a CNF in which all clauses have width $\leq w$.

The {\em resolution proof system} is a Hilbert-style proof system whose lines are clauses and that
has only one {\em resolution rule}
\begin{equation} \label{eq:res}
\binaryinf{C\lor x_i^a}{D\lor x_i^{1-a}}{C\lor D},\ a\in\{0,1\}.
\end{equation}
We will sometimes make use of the notation $\res(C\lor x_i^a, D\lor
x_i^{1-a})$.

The {\em size} of a resolution proof $\Pi$, denoted as $|\Pi|$, is the number
of lines in it. For a CNF $\tau$ and a clause $C$, $S_R(\tau\vdash C)$ is the
minimal possible size of a resolution proof of the clause $C$ from clauses in
$\tau$ ($\infty$ if $C$ is not implied by $\tau$). Likewise, $w(\tau\vdash
C)$ is the minimal possible width of such a proof, defined as the maximal
width of a clause in it. For a proof $\Pi$ that derives $C$ from $\tau$, the
clauses in $\tau$ that appear in $\Pi$ are called {\it axioms}, and if $C =
0$ then $\Pi$ is called a {\it refutation}. Let $\var(\Pi)$ denote the set of
variables appearing in $\Pi$, i.e.\ the union of $\var(C)$ for $C$ appearing
in $\Pi$.

Note that the {\em weakening rule}
$$
\unaryinf{C}{C\lor D}
$$
is {\em not} included by default. In the full system of resolution it is admissible in the sense that $S_R(\tau\vdash 0)$ does not change if we allow it. But this will not be the case for some of the CDCL-based fragments we will be considering below.

\begin{remark}\label{rem:weakening}
  Despite the above distinction, it is often convenient to consider systems that do allow the weakening rule.
  We make it clear when we do this by adding the annotation `$+$ weakening' to the system.
  For example, resolution $+$ weakening is the resolution proof system with the weakening rule included (as in Section \ref{sec:weakening}).
\end{remark}

\subsubsection*{Resolution Graphs}

Our results depend on the careful analysis of the structure of resolution proofs.
For example, it will be useful for us to maintain structural properties of the proof while changing the underlying clauses and derivations.
We build up the following collection of definitions for this analysis, to which we will refer throughout the later sections. The reader may skip this section for now and return to it in the future as needed.

\begin{definition}\label{df:graph}
For a resolution $+$ weakening proof $\Pi$, its {\it{resolution graph, $G(\Pi)$}}, is an acyclic directed graph representing $\Pi$ in
the natural way: each clause in $\Pi$ has a distinguished node, and for each node there are incoming edges from the nodes
corresponding to the clauses from which it is derived. Every node has in-degree 0, 1, or 2 if its corresponding clause is an axiom,
derived by weakening, or derived by resolving two clauses, respectively. Denote the set of nodes by $V(\Pi)$, and the clause
at $v\in V(\Pi)$ by $c_\Pi(v)$. We do {\em not} assume that $c_\Pi$ is injective, that is we allow the same clause to appear in the proof several times. There is a natural partial order on $V(\Pi)$ reflecting the order of appearances of clauses
in $\Pi$: $v>u$ if and only if $v$ is a descendant of $u$, or equivalently, there is a (directed) path
from $u$ to $v$. We sometimes say that
$v$ is {\it{above}}
(resp.\ {\it below}) $u$ if $v>u$ (resp.\ $v<u$). If, moreover, $(u,v)$ is an edge (directed from $u$ to $v$), we say that $u$ is
a {\em parent} of $v$. A set of nodes is {\it independent} if any two nodes in the set are incomparable.
 {\it{Maximal}}
and {\it{minimal}} nodes of any nonempty $S\subset V(\Pi)$ are defined with respect to this partial order: $\max_\Pi S \df \{
v\in S: \forall u\in S\neg(v<u)\}$, and similarly for $\min_\Pi S$.
\end{definition}

\begin{definition}
Let $S \subseteq V(\Pi)$.
The \textit{upward closure} and \textit{downward closure} of $S$ in $G(\Pi)$ are $\ucl_{\Pi}(S) \df \{v \in V(\Pi) :
\exists w \in S (v \geq w)\}$ and $\dcl_{\Pi}(S) \df \{v \in V(\Pi) : \exists w \in S (v \leq w)\}$, respectively.
A subset of nodes $S$ is {\it{parent-complete}} if for any $v\in S$ of in-degree 2, one parent of $v$ being in $S$ implies that the other parent of $v$ is also in $S$.
It is {\it{path-complete}} if for any directed path $p$ in $G(\Pi)$, the two end points of $p$ being in $S$ implies all nodes of $p$ are.
\end{definition}

\begin{example} \label{ex:complete}
The upward closure $\ucl_{\Pi}(S)$ is path-complete but need not be parent-complete.
The downward closure $\dcl_{\Pi}(S)$ is always both path-complete and parent-complete.
\end{example}

These definitions behave naturally, as demonstrated by the following proposition.

\begin{proposition}\label{prop:p1}
Let $S\subset V(\Pi)$ be a nonempty set of nodes that is both parent-complete and path-complete. Then the following holds.
\begin{enumerate}
\item \label{a} The induced subgraph of $G(\Pi)$ on $S$ is the graph of a
subproof in $\Pi$ of $\max_\Pi S$ from  $\min_\Pi S$;

\item \label{b} If $v\in S$ has a parent that is not in $S$ then $v\in\min_{\Pi}S$.
\end{enumerate}
\end{proposition}

\begin{proof}
\\
\ref{a}. Let $S^\ast\subseteq S$ be the set of all nodes in $S$ ``provable'' from $\min_\Pi S$ inside $S$. Formally, it is the closure of
$\min_\Pi S$ according to the following rule: if $v\in S$ and all its parents are in $S^\ast$ then $v$ is also in $S^\ast$. We need to
show that $S^\ast=S$.

Assume not, and fix an arbitrary $v\in\min_\Pi(S\setminus S^\ast)$. Since $v\not\in\min_\Pi S$, there exists $w\in S$ below $v$.
Since $S$ is path-complete, we can assume w.l.o.g. that $w$ is a parent of $v$, and since $S$ is parent-complete, all parents of
$v$ are in $S$. Now, since $v$ is {\em minimal} in $S\setminus S^\ast$, all of them must be actually in $S^\ast$. Hence
$v\in S^\ast$, a contradiction.
\medskip

\noindent \ref{b}. All parents of $v$ are not in $S$ by parent-completeness and, for all $u < v$, it follows that $u\not\in S$ by path-completeness.
\end{proof}

\begin{definition}
  A resolution graph is \textit{connected} if $|\max_\Pi V(\Pi)|=1$, i.e.  there is a unique sink.
\end{definition}

\begin{example}
  For a resolution proof $\Pi$ and $v \in V(\Pi)$, the subgraph on $\dcl_\Pi(\{v\})$ is a connected resolution graph whose axiom nodes are among those of $G(\Pi)$.
\end{example}

\subsubsection*{Ordered and Half-Ordered Resolution}

Fix now an order $\pi\in S_n$.
For any literal $l=x_k^{a}$, $\pi(l)\df\pi(k)$.
For $k \in [n]$, let $\var_\pi^k$ denote the $k$ smallest variables according to $\pi$.
Say that a clause $C$ is \textit{$k$-small} with respect to $\pi$ if $\var(C) \subseteq \var_\pi^k$.

The proof system {\em $\pi$-ordered resolution} is the subsystem of resolution defined by imposing the following restriction on the resolution rule \eqref{eq:res}:
$$
\forall l \in C\lor D\ (\pi(l) < \pi(x_i)).
$$
That is, the two antecedents are $i$-small. We note that in the literature this system is usually defined differently, namely in a top-down manner (see e.g.
\cite{2000-Bonet-Esteban-Galesi-Johannsen}). It is easy to see, however, that our version is equivalent.

\begin{definition}
{\em $\pi$-half-ordered resolution} is the subsystem of resolution in which the rule \eqref{eq:res} is restricted by the requirement
\begin{equation} \label{eq:half-ordered}
\forall l\in C\ (\pi(l)<\pi(x_i)).
\end{equation}
That is, at least one of the antecedents is $i$-small.
\end{definition}

Recall \cite{1979-Cook-Reckhow} that a proof system $P$ $p$-simulates another proof system $Q$ if there exists a polynomial time algorithm that takes any $Q$-proof to a $P$-proof from the same axioms (in particular, the size of the $P$-proof is bounded by a polynomial in the size of the original proof). Two systems $P$ and $Q$ are {\em polynomially equivalent} if they $p$-simulate each other.

We are now ready to state our first result.

\begin{theorem} \label{thm:p1}
For any order $\pi\in S_n$, $\pi$-ordered resolution is polynomially equivalent to $\pi$-half-ordered resolution.
\end{theorem}

The next proof system, $P_0$, is even more heavily motivated by CDCL solvers. For this reason we interrupt our proof-complexity exposition to define the corresponding model.

\subsection{CDCL-based proof systems} \label{sec:cdcl}

As we noted in Introduction, we will try to highlight certain subtle points in the definition of the model by injecting informal remarks.

A {\em unit clause} is a clause consisting of a single literal.
An {\em assignment} is an expression of the form $x_i=a\ (1\leq i\leq n,\ a\in\{0,1\})$.
A {\em restriction} $\rho$ is a set of assignments in which all variables are pairwise distinct.
We denote by $\var(\rho)$ the set of all variables appearing in $\rho$.
Restrictions naturally act on clauses, CNFs and resolution proofs, etc.; we denote by $C|_\rho,\ \tau|_\rho,\ \Pi|_\rho\ldots$ the result of this action. Note that both $\pi$-ordered resolution and $\pi$-half-ordered resolution are closed under restrictions, i.e.\ if $\Pi$ is a $\pi$-(half)-ordered resolution proof, then $\Pi|_\rho$ is a $\pi|_\rho$-(half)-ordered resolution proof of no-bigger size, where $\pi|_\rho$ is the order induced by $\pi$ on $V\setminus \var(\rho)$.

An {\em annotated assignment} is an expression of the form $x_i \stackrel * = a\ (1\leq i\leq n,\ a\in\{0,1\}, * \in \{d, u\})$.
Informally, a CDCL solver builds (ordered) restrictions one assignment at a time, and the annotation indicates in what way the assignment is made: `$d$' means by a decision, and `$u$' means by unit propagation.
See Definition \ref{df:action} and Remark \ref{rem:annotations} below for details about these annotations.

\begin{definition} \label{df:trail}
A {\em trail} is an ordered list of annotated assignments in which all
variables are again pairwise distinct. A trail acts on clauses, CNFs etc.\
just in the same way as does the restriction obtained from it by disregarding
the order and the annotations on assignments. For a trail $t$ and an
annotated assignment $x_i \stackrel * = a$ such that $x_i$ does not appear in
$t$, we denote by $[t,x_i \stackrel * = a]$ the trail obtained by appending
$x_i \stackrel * = a$ to its end. $t[k]$ is the $k$th assignment of $t$.  A
{\em prefix} of a trail $t=[x_{i_1} \stackrel{*_1}= a_1,\ldots, x_{i_r}
\stackrel{*_r}= a_r]$ is any trail of the form $[x_{i_1} \stackrel{*_1}= a_1,
\ldots, x_{i_s} \stackrel{*_s}= a_s]\ (0\leq s\leq r)$ denoted by $t[\leq
s]$. $\Lambda$ is the empty trail.

A {\em state} is a pair $(\mathbb C, t)$, where $\mathbb C$ is a CNF and $t$ is a trail. The state $(\mathbb C, t)$ is
{\em terminal} if either $C|_t\equiv 1$ for all $C\in\mathbb C$ or $\mathbb C$ contains 0. All other states are nonterminal. We let $\mathbb S_n$ denote the set of all states (recall that $n$ is reserved for the number of variables), and let $\mathbb S_n^{\text o}\subset \mathbb S_n$ be the set of all nonterminal states.
\end{definition}

\begin{remark} \label{rem:positional}
As unambiguous as Definition \ref{df:trail} may seem, it already reflects one important choice,  to consider only  {\em positional}\footnote{The name is suggested by a similar term ``positional strategy'' in game theory.} solvers, i.e. those that are allowed to carry along only CNFs and trails, but not any other auxiliary information. The only mathematical ramification of this restriction is that we will have to collapse the whole clause learning stage into one step, but that is a sensible thing to do anyway. From the practical perspective, however, this restriction is far from obvious and we will revisit this issue in our concluding section \ref{sec:concl}.
\end{remark}

\begin{remark}
We are now about to describe the core of our (or, for that matter, any other) model, that is transition rules between states. But since this definition is the longest one, we prefer to change gears and {\em precede} it with some informal remarks rather than give them after the definition.

Proof systems attempting to capture performance of modern CDCL solvers are in general much bulkier than their logical counterparts and are built from several heterogeneous blocks. At the same time, most papers highlight the impact of one or a few of the features, with a varying degrees of nondeterminism allowed, while the features out of focus are treated in often unpredictable and implicit ways. We have found this state of affairs somewhat impending for the effort of trying to compare different results to each other or to build useful structure around them of the kind existing in ``pure'' proof complexity. Therefore, we adapt an approach that in a sense is the opposite. Namely, we rigorously describe a {\em basic} model that is very liberal and nondeterministic and intends to approximate the union of most conceivable features of CDCL solvers. Then models of actual interest will be defined by their {\em deviations} from the basic model. These deviations will take the form of ``amendments'' forbidding certain forms of behavior or, potentially, allowing for new ones.

Besides this point, there are only few (although sometimes subtle) differences from the previous models, so our description is given more or less matter-of-factly.
\end{remark}

\begin{definition}\label{df:action}
For a (nonterminal) state $S = (\mathbb C,t)\in \mathbb S_n^{\text o}$, we define the finite nonempty set $\text{Actions}(S)$ and the function $\text{Transition}_S\function{\text{Actions}(S)}{\mathbb S_n}$; the fact
$\text{Transition}_S(A)=S'$ will be usually abbreviated to $S\stackrel{A}{\Longrightarrow} S'$. Those are described as follows:
$$
\text{Actions}(S) \df D(S) \stackrel . \cup U(S) \stackrel .\cup L(S),
$$
where the letters $D,U,L$ have the obvious meaning\footnote{Restarts will be treated as a part of the learning scheme.}.

\begin{itemize}
\item $D(S)$ consists of all annotated assignments $x_i \stackrel d = a$ such that $x_i$ does not appear in $t$ and $a\in\{0,1\}$. We naturally let
    \begin{equation} \label{eq:decision}
    (\mathbb C,t)\stackrel{x_i \stackrel d = a}{\Longrightarrow}(\mathbb C, [t, x_i \stackrel d = a]).
    \end{equation}

\item $U(S)$ consists of all those assignments $x_i \stackrel u = a$ for which $\mathbb C|_t$ contains the unit clause $x_i^a$; the transition function is given by the same formula \eqref{eq:decision} but with a different annotation:
\begin{equation} \label{eq:unit-prop}
(\mathbb C,t)\stackrel{x_i \stackrel u = a}{\Longrightarrow}(\mathbb C, [t, x_i \stackrel u = a]).
\end{equation}

\item As should be expected, $L(S)$ is the most sophisticated part of the
    definition (cf. \cite[Section 2.3.3]{2011-Atserias-Fichte-Thurley}).
    Let $t = [x_{i_1} \stackrel{*_1}= a_1,\ldots,x_{i_r} \stackrel{*_r}=
    a_r]$. By reverse induction on $k=r+1,\ldots,1$ we define the set
    $\mathbb C_k(S)$ that, intuitively, is the set of clauses that can be
    learned by backtracking up to the prefix $t[\leq k]$.

We let
\[
    \mathbb C_{r+1}(S) \df \set{D\in\mathbb C}{D|_t=0}
\]
be the set of all {\em conflict clauses}.

For $1\leq k \leq r$, we do the following:
if the $k$-th assignment of $t$ is of the form $x_{i_k} \stackrel d = a_k$, then $\mathbb{C}_k(S) \df \mathbb{C}_{k + 1}(S)$.
Otherwise, it is of the form $x_{i_k} \stackrel u= a_k$, and we build up $\mathbb C_k(S)$ by processing every clause $D\in\mathbb C_{k+1}(S)$ as follows.
\begin{itemize}

\item If $D$ does not contain the literal $\overline{x_{i_k}^{a_k}}$ then we include $D$ into $\mathbb C_k(S)$ unchanged.

\item If $D$ contains $\overline{x_{i_k}^{a_k}}$, then we resolve $D$
    with all clauses $C \in \mathbb C$ such that $C|_{t[\leq k-1]}= x_{i_k}^{a_k}$ and
    include into $\mathbb C_k(S)$ all the results $\res(C,D)$. $D$ itself is not
    included.
\end{itemize}

To make sure that this definition is sound, we have to guarantee that $C$ and $D$ are actually resolvable (that is, they do not contain any other conflicting variables but $x_{i_k}$). For that we need the following observation, easily proved by reverse induction on $k$, simultaneously with the definition:
\begin{claim} \label{clm:c_property}
$D|_t=0$ for every $D\in\mathbb C_k(S)$.
\end{claim}

Finally, we let
$$
\mathbb C(S)\df \bigcup_{k=1}^r\mathbb C_k(S),
$$
\begin{equation} \label{eq:learning}
L(S)\df
\begin{cases}
\{(0,\Lambda)\}\ \text{if}\ 0\in \mathbb C(S);\\
\set{(C,t^\ast)}{C\in \of{\mathbb C(S)\setminus \mathbb C},\ t^\ast\ \text{a prefix of}\ t\ \text{such that}\  C|_{t^\ast} \neq 0}\ \text{otherwise}
\end{cases}
\end{equation}
and
$$
(\mathbb C,t) \stackrel{(C,t^\ast)}{\Longrightarrow} (\mathbb C\cup \{C\}, t^\ast).
$$
\end{itemize}
This completes the description of the basic model.
\end{definition}

\begin{remark} \label{rem:annotations}
For nearly all modern implementations of CDCL, the annotations are redundant because CDCL solvers typically require unit propagation always to be performed when it is applicable (in our language of amendments, this feature will be called {\sf ALWAYS-U}).
Nevertheless, the presence of annotations makes the basic model flexible enough to carry on various, sometimes subtle, restrictions and extensions. In particular, we consider solvers that are not required to {\em record} unit propagations as such.
This allows for the situation in which $x_i \stackrel d = a$ {\em and} $x_i \stackrel u = a$ are in $\text{Actions}(S)$, and the set of learnable clauses is sensible to this.
\end{remark}

\begin{remark}
In certain pathological cases, mostly resulting from neglecting to {\em do} unit propagation, the set $\text{Actions}(\mathbb C,t)$ may turn out to be empty
even if $(\mathbb C,t)$ is nonterminal and $\mathbb C$ is contradictory. But for the reasons already discussed above, we prefer to keep the basic model as clean as possible {\em syntactically}, postponing such considerations for later.
\end{remark}

The {\em transition graph} $\Gamma_n$ is the directed graph on $\mathbb S_n$ defined by erasing the information about actions; thus $(S,S') \in E(\Gamma_n)$ if and only if $S'\in\text{im}(\text{Transition}_S)$. It is easy to see (by double induction on $(|\mathbb C|, n - |t|)$) that $\Gamma_n$ is acyclic. Moreover, both the set
$\set{(S,A)}{A\in \text{Actions}(S)}$ and the function $(S,A)\mapsto \text{Transition}_S(A)$ are polynomially\footnote{in the size of the state $S$, not in $n$} time computable. These observations motivate the following definition.

\begin{definition} \label{df:solver}
Given a CNF $\mathbb C$, a \textit{partial run} on $\mathbb C$ from the state $S$ to the state $T$ is a sequence
\begin{equation} \label{eq:run}
S = S_0 \stackrel{A_0}{\Longrightarrow} S_1\stackrel{A_1}{\Longrightarrow}\ldots S_{L-1}\stackrel{A_{L-1}}{\Longrightarrow} S_L = T,
\end{equation}
where $A_k\in \text{Actions}(S_k)$.
In other words, a partial run is an annotated path in $\Gamma_n$.
A \textit{successful run} is a partial run from $(\C, \Lambda)$ to a terminal state.
A {\em CDCL solver} is a {\bf partial} function $\mu$
on $\mathbb S_n^{\text o}$ such that $\mu(S)\in \text{Actions}(S)$ whenever $\mu(S)$ is defined. The above remarks imply that
when we apply a CDCL solver $\mu$ to any initial state $(\mathbb C,\Lambda)$, it will always result in a
finite sequence like \eqref{eq:run}, with $T$ being a terminal state (successful run) or such that $\mu(T)$
is undefined (failure).
\end{definition}

\begin{remark}
Theoretical analysis usually deals with {\em classes} (i.e., sets) of individual solvers rather
than with individual implementations, and there might be several different approaches to defining such classes.
One might consider for example various complexity restrictions like demanding that $\mu$ be polynomially time
computable.  But in this paper we are more interested
in classes defined by prioritizing and restricting various actions.
\end{remark}

\begin{definition} \label{df:local}
A {\em local} class of CDCL solvers is described by a collection of subsets $\text{AllowedActions}(S)\subseteq
\text{Actions}(S),\ S\in\mathbb S_n^{\text o}$. It consists of all those solvers $\mu$ for which $\mu(S)\in \text{AllowedActions}(S)$, whenever $\mu(S)$ is defined.
\end{definition}

We will describe local classes of solvers in terms of {\em amendments} prescribing what actions should be
{\em removed} from the set $\text{Actions}(S)$ to form $\text{AllowedActions}(S)$.
 Without further ado, let us give a few examples illustrating
how do familiar restrictions look in this language. Throughout the description, we fix a nonterminal state
$S = (\mathbb C,t)$.

\begin{description}
\item[\sf{ALWAYS-C}] If $\mathbb C|_t$ contains the empty clause, then $D(S)$ and $U(S)$ are removed from
$\text{Actions}(S)$. In other words, this amendment requires the solver to perform conflict analysis if it
can do so.

\item[{\sf ALWAYS-U}] If $\mathbb C|_t$ contains a unit clause, then $D(S)$ is removed from
$\text{Actions}(S)$. This amendment insists on unit propagation, but leaves to nondeterminism the choice of the unit to propagate if there are several choices. Note that as defined, {\sf ALWAYS-U} is a
lower priority amendment than {\sf ALWAYS-C}: if both a conflict and a unit clause are present, the solver
must do conflict analysis.

\item[{\sf ALWAYS-R}] In the definition \eqref{eq:learning} of learning actions we keep only those $(C,t^\ast)$
for which $t^\ast = \Lambda$.

\item[{\sf NEVER-R}] In the definition \eqref{eq:learning}, we require that $t^\ast$ is the {\em longest}
prefix of $t$ satisfying $C|_{t^\ast}\neq 0$ (in which case $C|_{t^\ast}$ is necessarily a unit clause).
As described, this amendment does not model nonchronological backtracking or require that the last assignment in the trail is a decision.
However, this version is easier to state and it is not difficult to modify to have the aforementioned properties.
Furthermore, all open questions pertaining to this amendment remain open for either version.

\item[{\sf ASSERTING-L}] In the definition \eqref{eq:learning}, we shrink $\mathbb C(S) \setminus \mathbb C$ to $\of{\bigcup_{k=1}^s \mathbb C_k(S)} \setminus \mathbb C$, where $s<r$ is the largest index for which $x_{i_s}=a_s$ is annotated as `$d$' in $t$.
This amendment is meaningful (and mostly used) only when combined with {\sf ALWAYS-C} and {\sf ALWAYS-U}, in which case we can state expected properties like the fact that every learned clause contains the literal $x_{i_s}^{1 - a_s}$.

\item[$\decisionl$] In the definition \eqref{eq:learning}, we shrink $\mathbb C(S) \backslash \mathbb C$ to
$\mathbb C_1(S) \backslash \mathbb C$.

\item[{\sf FIRST-L}] In the definition \eqref{eq:learning}, we shrink $\mathbb C(S)
\setminus \mathbb C$ to those clauses that are obtained by resolving, in the notation of Definition \ref{df:action}, between pairs
$C$ and $D$ with $D\in\mathbb C$. As noted in the introduction, this is similar to the scheme FirstNewCut \cite{2004-Beame-Kautz-Sabharwal}.

\item[\protect{$\pi$-{\sf D}, where $\pi\in S_n$ is an order on the variables}] We keep in $D(S)$ only the two assignments
$x_i \stackrel d = 0$, $x_i \stackrel d= 1$, where $x_i$ is the {\em smallest} variable w.r.t. $\pi$ that does not appear in $t$. Note that this amendment does not have any effect upon $U(S)$, and the main technical contributions of our paper
can be also phrased as asking under which circumstances this ``loophole" can circumvent the severe restriction placed on the set $D(S)$.

\item[{\sf WIDTH}-$w$, where $w$ is an integer] In the definition \eqref{eq:learning}, we keep in $\mathbb C(S) \setminus \mathbb C$ only clauses of width $\leq w$. Note that
this amendment still allows us to use wide clauses as intermediate results {\em within} a single
clauses learning step.

\item[{\sf SPACE}-$s$, where $s$ is an integer] If $|\mathbb C|\geq s$, then $L(S)$ is entirely
removed from $\text{Actions}(S)$. This amendment makes sense when accompanied by the possibility to
do bookkeeping by removing ``unnecessary" clauses. We will briefly discuss positive amendments in
Remark \ref{rem:positive} below.
\end{description}

Thus, our preferred way to specify local classes of solvers and the corresponding proof systems is by
listing one or more amendments, with
the convention that their effect is cumulative: an action is removed from $\text{Actions}(S)$ if and
only if it should be removed according to at least one of the amendments present. More formally,

\begin{definition}
For a finite set $\mathcal A_1,\ldots,\mathcal A_r$ of poly-time computable amendments, we let
$\text{{\sf CDCL}}(\mathcal A_1,\ldots,\mathcal A_r)$ be the (possibly incomplete) proof system whose proofs
are those successful runs \eqref{eq:run} in which none of the actions $A_i$ is affected by any of the
amendments $\mathcal A_1,\ldots,\mathcal A_r$.
\end{definition}

\begin{remark}
The amendments {\sf ALWAYS-C}, {\sf ALWAYS-U} are present in almost all previous work and, arguably, it is
precisely what distinguishes conflict-driven clause learning techniques. Nonetheless, we have decided
against including them into the basic model as they may be distracting in theoretical studies focusing
on other features; our work is one example. Also, from the practical point of view (this sentence is
admittedly rather hypothetical and speculative!) it is not \textit{a priori} clear why it is always a good idea
to jump up at the very first opportunity to do unit propagation or conflict analysis instead of waiting
for ``more promising'' clauses to pop up.
\end{remark}

\begin{remark} \label{rem:positive}
Let us briefly discuss the possibility of \textit{extending} the basic model rather than restricting it.
The most substantial deviation would be to forfeit the assumption of positionality (see Remark \ref{rem:positional}) or, in other words, to allow the solver to carry along more information than
just a set of clauses and a trail.
Two such examples are \textit{dynamic variable ordering} and \textit{phase saving}.
The first is very pertinent to the technical part of our paper, so we defer the corresponding discussion to Section \ref{sec:concl}.

For positional solvers, extending the basic model amounts to introducing {\em positive amendments}
enlarging the sets $\text{Actions}(S)$ instead of decreasing them. Here are a few suggestions we came
across during our deliberations.
\begin{description}
\item[{\sf CLAUSE DELETION}] For $S = (\mathbb C,t)\in \mathbb S_n^{\text o}$, we add to $\text{Actions}(S)$
all subsets $\mathbb C_0\subseteq \mathbb C$. The transition function is obvious:
$$
(\mathbb C, t) \stackrel{\mathbb C_0}{\Longrightarrow} (\mathbb C_0, t).
$$
This is the space model whose study was initiated in \cite{2016-Elffers-et-al}, and like in that paper, we do not see compelling
reasons to differentiate between original clauses and the learned ones.

\item[{\sf MULTI-CLAUSE LEARNING}] In the definition \eqref{eq:learning} of $L(S)$, we can allow arbitrary
nonempty subsets $\mathbb C_0\subseteq \mathbb C(S)\setminus \mathbb C$ instead of
a single clause $C$ {and require that $C|_{t^\ast}\neq 0$} for {\em any} $C\in\mathbb C$, with the obvious transition
$$
(\mathbb C,t) \stackrel{(\mathbb C_0,t^\ast)}{\Longrightarrow} (\mathbb C\cup \mathbb C_0,t^\ast).
$$
Though existing SAT-solver implementions tend not to do this, it is natural to consider when thinking of Pool resolution or RTL proof systems as variants of CDCL.

\item[{\sf INCOMPLETE LEARNING}] In the definition \eqref{eq:learning} of $L(S)$, we could remove the
restriction $C|_{t^\ast}\neq 0$ on the prefix $t^\ast$. This positive amendment
could make sense in the absence of {\sf ALWAYS-C}, that is, if we are prepared for delayed conflict analysis.
\end{description}
\end{remark}

\bigskip
In this language, the (nonalgorithmic part of the) main result from \cite{2011-Atserias-Fichte-Thurley,2011-Pipatsrisawat-Darwiche} can be roughly summarized as
\begin{displayquote}
{\sf CDCL}({\sf ALWAYS-C}, {\sf ALWAYS-U}, {\sf ALWAYS-R}, $\decisionl$) \it is polynomially equivalent to resolution.
\end{displayquote}
The algorithmic part from \cite{2011-Atserias-Fichte-Thurley} roughly says
that {\em any} CDCL solver in the associated class, subject to the only
condition that the choice of actions from $D(S)$ (when it is allowed by the
amendments) is random, polynomially simulates bounded-width
resolution\footnote{That is, has running time $n^{O(w(\tau_n\vdash 0))}$ with
high probability, given an contradictory CNF $\tau_n$ as an input.}. The open
question asked in \cite[Section 2.3.4]{2011-Atserias-Fichte-Thurley} can be
reasonably interpreted as whether {\sf CDCL}({\sf ALWAYS-C}, {\sf ALWAYS-U},
{\sf WIDTH}-$w$) is as powerful as width-$w$ resolution, perhaps with some
gap between the two width constraints (We took the liberty to remove those
amendments that do not appear to be relevant to the question.) Finally, we
would like to abstract the ``no-restarts" question as

\begin{displayquote}
\it Does {\sf CDCL}{\rm (}{\sf ALWAYS-C}, {\sf ALWAYS-U}, {\sf NEVER-R}{\rm )} {\rm (}or at least {\sf CDCL}{\rm (}{\sf NEVER-R}{\rm ))} simulate general resolution?
\end{displayquote}
where we have again removed all other amendments in the hope that this will make the question more clean mathematically.

\subsection{Our contributions}

As they had already been discussed in the introduction, here we formulate our results (in the language just introduced) more or less matter-of-factly.

\begin{theorem} \label{thm:last-l}
For any fixed order $\pi$ on the variables, the system {\sf CDCL}{\rm (}$\decisionl$, $\pi$-{\sf D}{\rm )} is polynomially
equivalent to $\pi$-ordered resolution.
\end{theorem}

\begin{theorem} \label{thm:main}
For any fixed order $\pi$ on the variables, the system {\sf CDCL}{\rm (}{\sf FIRST-L}, $\pi$-{\sf D}{\rm )} is polynomially
equivalent to general resolution.
\end{theorem}

\begin{theorem} \label{thm:width}
For any fixed order $\pi$ on the variables and every $\epsilon>0$ there exist contradictory CNFs $\tau_n$
with $w(\tau_n\vdash 0)\leq O(1)$ not provable in {\sf CDCL}{\rm (}$\pi$-\sf D, {\sf WIDTH}-$(1-\epsilon)n${\rm )}.
\end{theorem}

\medskip
Finally, let us mention that while $\text{{\sf CDCL}}(\mathcal A_1,\ldots, \mathcal A_r)$ is a (possibly incomplete) proof system
in the Cook-Reckhow sense, it need not necessarily be a Hilbert-style proof system,
operating with ``natural" lines and inference rules. Assume, however, that the set $\text{AllowedActions}(S)$
additionally satisfies the following two properties:
\begin{enumerate}
\item whenever $\text{AllowedActions}(S)\cap L(S)\neq\emptyset$, it contains an action leading to a state of the form
$(C, \Lambda)$ (i.e, restarts are allowed);

\item (monotonicity) If $S=(\mathbb C,t)$, $S'= (\mathbb C', t)$ and $\mathbb C\subseteq \mathbb C'$ then
$\text{AllowedActions}(S)\cap (D(S)\stackrel . \cup U(S))\subseteq \text{AllowedActions}(S')\cap (D(S') \stackrel{.}{\cup} U(S'))$.
\end{enumerate}

Then every trail $t$ that appears in a run can always be {\em recreated}, at a low cost, when it is needed again.
Thus, under these restrictions we get a ``normal" proof system with nice properties. We formulate it explicitly
for the case $\pi$-{\sf D} we are mostly interested in.

\begin{definition}\label{df:p0}
Fix an order $\pi$ on the variables. $\pi$-$P_0$ is the following (two-typed) proof system. Its
lines are either clauses or trails, and it has the following rules of inference:
$$
\unaryinf t{[t,x_i\stackrel d=a]},\ \ \ \text{(Decision rule)}
$$
where $x_i$ is the $\pi$-smallest index such that $x_i$ does not appear in $t$ and $a\in\{0,1\}$ is arbitrary;

$$
\binaryinf{t}{C}{[t,x_i\ueq a]},\ \ \ \text{(Unit propagation rule)}
$$
where $C|_t = x_i^a$;

$$
\ternaryinf{C\lor x_i^a}{D\lor x_i^{1-a}}{t}{C\lor D},\ \ \ \text{(Learning rule)}
$$
where $(C\lor D)|_t = 0$, $(x_i\stackrel *=a)\in t$ and all other variables of $C$ appear before
$x_i$ in $t$.
\end{definition}

It is straightforward to see that without the unit propagation rule, this is just the $\pi$-half-ordered resolution.

Then, the main technical part in proving Theorem \ref{thm:main} is the following

\begin{theorem} \label{thm:p0}
For every fixed order $\pi$ on the variables, $\pi$-$P_0$ is polynomially equivalent to general resolution.
\end{theorem}

%% file: last_l.tex
In this section we prove Theorem \ref{thm:last-l}.
The proof breaks into two parts (Theorem \ref{thm:p1}, Theorem \ref{thm:cdcl=p1}), with half-ordered resolution as the intermediate.

\subsection{ Half-Ordered $=_p$ Ordered} \label{sec:proofcomplexity}

Half-ordered resolution trivially $p$-simulates ordered resolution, so the core of Theorem \ref{thm:p1} is the other direction.
In this section we will depend heavily on resolution graphs (Definition \ref{df:graph}) and related definitions from Section \ref{sec:prel}.

\begin{definition}
A resolution refutation $\Pi$ is \textit{ordered up to $k$} (with respect to
an order $\pi$) if it satisfies the property that $\forall v\in V(\Pi)$, if
$c_\Pi(v)$ is derived by resolving two clauses on some variables $x_i \in
\var_\pi^k$, then all resolution steps above $v$ are on variables in
$\var_\pi^{k - 1}$. We note that $\pi$-ordered resolution proofs are
precisely those that are ordered up to $n-1$.
\end{definition}

We now prove the main part of Theorem \ref{thm:p1}, namely that $\pi$-ordered resolution $p$-simulates $\pi$-half-ordered resolution.

\begin{proof}{\it{(of theorem \ref{thm:p1})}}\
Let $\Pi$ be a $\pi$-half-ordered resolution refutation of $\tau$.
Without loss of generality, assume that $\pi=\text{id}$ (otherwise rename variables).

We will construct by induction on $k$ ($0\leq k\leq n-1$) a half-ordered
resolution refutation $\Pi_k$ of $\tau$, which is ordered up to $k$. For the
base case, let $\Pi_0 = \Pi$. Suppose $\Pi_k$ has been constructed; without
loss of generality we can assume that $\Pi_k$ is connected (otherwise take
the subrefutation below any occurrence of 0).

Consider the set of nodes whose clauses are $k$-small. Note that this set is
parent-complete. We claim that it is also upward-closed (and hence
path-complete). Indeed, let $u$ be a parent of $v$, and assume that $c(u)=c_{\Pi_k}(u)$ is
$k$-small. Then (since we disallowed weakenings), $c(v)$ is obtained by
resolving on a variable $x_i \in \var_\pi^k$. But since $\Pi_k$ is ordered up
to $k$, $\var(c(v)) \subseteq \var_\pi^{i - 1} \subseteq \var_\pi^k$, as
otherwise some variable in $c(v)$ would have remained unresolved on a path
connecting $v$ to the sink (here we have used connectedness). Hence $c(v)$ is
also $k$-small.

So by Proposition \ref{prop:p1}(\ref{a}), it
defines the subrefutation of the independent set
\begin{equation}\label{eq:$L_i$}
L_k \df {\min}_{\Pi_k} \{v |\ \text{$c(v)$ is $k$-small}\}.
\end{equation}
Since any node in $\Pi_k$ can be connected to the sink (which is $k$-small)
by a path, and the set of $k$-small nodes is parent-complete, by Proposition
\ref{prop:p1}, $L_k$ splits $\Pi_k$ into two parts:
$\Pi_k=\ucl_{\Pi_k}(L_k)\cup \dcl_{\Pi_k}(L_k)$, $L_k=\ucl_{\Pi_k}(L_k)\cap
\dcl_{\Pi_k}(L_k)$, where $\dcl_{\Pi_k}(L_k)$ is (the graph of) a subproof of
$L_k$ and $\ucl_{\Pi_k}(L_k)$ is (the graph of) a subrefutation of $L_k$.
Denote these two subproofs by $D$ and $U$, respectively: $D$ is comprised of
all nodes in $\Pi$ that either are marked by a clause that is not $k$-small
or belong to $L_k$, and $U$ is comprised of all nodes marked by an $k$-small
clause. In particular, all axioms are in $D$, all resolutions in $U$ are on
the variables in $\var_\pi^k$ and, since $\Pi_k$ is ordered up to $k$, all
resolutions in $D$ are on the variables not in $\var_\pi^k$. Let
\begin{equation}\label{eq: $M$}
M \df {\min}_{D} \{w |\ \text{$c(w)$ is the result of resolving two clauses on $x_{k+1}$}\}.
\end{equation}

If $M$ is empty, $\Pi_{k+1}\df \Pi_k$.

Otherwise, suppose $M=\{w_1,\dots,w_s\}$ and let

\begin{equation}\label{eq: $A_j$}
A_i \df \ucl_{D}(\{w_i\}).
\end{equation}
We will eliminate all resolutions on $x_{k + 1}$ in $D$ by the following
process, during which {\em the set of nodes stays the same} while edges and
clause-labeling function will possibly change. More precisely, we update $D$
in $s$ rounds, defining $\pi$-half-ordered resolution $+$ weakening proofs
$D_1, D_2, \dots, D_s$. Initially $D_0 = D$, $i=1$. Let $c_{i-1}$ denote the
clause-labeling $c_{D_{i-1}}$. To define the transition $D_{i-1}\to D_i$, we
need the following structural properties of $D_{i-1}$ (that will also be
proved by induction simultaneously with the definition).

\begin{claim}\label{cl:p1,1}
\hfill

\begin{enumerate}[label=\alph*., ref=\alph*]
\item\label{c1:1} Compared to $D$, $D_{i-1}$ has no new relation of `one node being above another' (there might be fewer);
\item\label{c1:2} $\forall v\in V(D)$, $c_{i-1}(v)$ is either $c_{D}(v)$ or $c_{D}(v)\vee x_{k+1}$ or $c_{D}(v) \vee \overline{x_{k + 1}}$;
\item\label{c1:3} If $v\notin \mathop\cup\limits_{j=1}^{i-1} A_j$ then $c_{i-1}(v) = c_D(v)$ and, moreover, this clause is obtained in $D_{i-1}$ with the same resolution rule as in $D$;
\item\label{c1:4} $D_{i-1}$ is a $\pi$-half-ordered resolution $+$
    weakening proof.
\end{enumerate}
\end{claim}
In the base case ($i=1$), Claim \ref{cl:p1,1} holds simply because $D_0=D$.

Let us construct $D_i$. By Claim \ref{cl:p1,1}\eqref{c1:3}, the resolution step at $w_i$ (which is not in $\mathop\cup\limits_{j=1}^{i-1} A_j$) is unchanged from $D$ to $D_{i-1}$. Assume that it resolves $c_{D}(w')=B \vee x_{k + 1}$ and $c_{D}(w'')=C \vee \overline{x_{k + 1}}$. Since $\Pi_k$ is half-ordered, either $B$ or $C$ is $k$-small. Assume without loss of generality that $B$ is $k$-small.

Recall that there is no resolution in $D$ on variables in $\var_\pi^k$. Thus, for all $v\in A_i$, it follows that $B$ is a subclause of $c_D(v)$, and by Claim \ref{cl:p1,1}\eqref{c1:2}, we get the following crucial property:
\begin{equation}\label{eq:subclause}
\text{For all $v \in A_i$, $B$ is a subclause of $c_{i-1}(v)$.}
\end{equation}
Note that $A_i$ is upward closed in $D_{i-1}$ by Claim \ref{cl:p1,1}\eqref{c1:1}. Accordingly, as the first step, for any $v\not\in A_i$ we set $c_i(v)\coloneqq c_{i-1}(v)$ and we do not touch its incoming edges.

Next, we update vertices $v\in A_i$ in an arbitrary $D$-topological order maintaining the property $c_i(v)\in \{c_{i-1}(v),\ c_{i-1}(v) \lor \overline{x_{k+1}}\}$ (in particular, $c_i(v)=c_{i-1}(v)$ whenever $c_{i-1}(v)$ contains a literal of the variable $x_{k+1}$).

First we set $c_{i}(w_i) \coloneqq c_{i-1}(w_i) \vee \overline{x_{k + 1}}$ (recall that $c_{D_{i-1}}(w_i)=c_D(w_i)$ by Claim \ref{cl:p1,1}\eqref{c1:3} and hence does not contain $x_{k+1}$ by \eqref{eq: $M$}), and replace incoming edges by a weakening edge from $w''$.

For $v\in A_i\setminus \{w_i\}$, we proceed as follows.
\begin{enumerate}
    \item
        If $x_{k + 1} \in c_{i-1}(v)$, keep the clause but replace incoming edges with a weakening edge $(w', v)$. This is well-defined by \eqref{eq:subclause}, and note for the record that since $w'<_D w <_D  v$, we do not enlarge the ``above'' relation compared to $D$.

    \item
        If $c_{i-1}(v)= \res(c_{i-1}(u), c_{i-1}(w))$ on $x_{k+1}$ where $\overline{x_{k + 1}} \in c_{i - 1}(u)$, set $c_{i}(v) \coloneqq c_{i-1}(v)\lor \overline{x_{k+1}}$ (or, equivalently, $c_{i-1}(v) \vee c_i(u))$ and replace incoming edges by a weakening edge $(u, v)$.

    \item
        If $c_{i-1}(v)$ is weakened from $c_{i-1}(u)$ (and $x_{k+1}\not\in c_{i-1}(v)$), set $c_{i}(v) \coloneqq c_{i-1}(v) \vee c_{i}(u)$. In other words, we append the literal $\overline{x_{k+1}}$ to $c_i(v)$ if and only if this was previously done for $c_i(u)$.

    \item\label{last_step}
        Otherwise, $x_{k+1}\notin c_{i-1}(v)$, and $c_{i - 1}(v) = \res(c_{i - 1}(u), c_{i - 1}(w))$ on some $x_\ell$, $\ell > k + 1$. In particular,
        $x_{k+1}\notin \{c_{i - 1}(u), c_{i - 1}(w)\}$. Set $c_{i}(v)\coloneqq\res(c_{i}(u), c_{i}(w))$ that is, like in the previous item, we append $\overline{x_{k+1}}$ if and only if it was previously done for either $c_i(v)$ or $c_i(w)$. Note that since $\ell>k+1$, this step remains $\pi$-half-ordered.
\end{enumerate}

This completes our description of $D_i$, but we still have to check Claim \ref{cl:p1,1} for it.
For \eqref{c1:1}, note that the only new edges were added in item 1, and see the remark made there. The cases $\eqref{c1:2}$ and $\eqref{c1:3}$ are straightforward. For $\eqref{c1:4}$, the only new resolution rules were introduced in item \ref{last_step}; again, see the remark made there.

The next claim summarizes the necessary properties of the end result, $D_s$.
\begin{claim}\label{cl:p1}
\hfill

\begin{enumerate}[label=\alph*., ref=\alph*]
  \item\label{c2:1} $D_s$ is a $\pi$-half-ordered resolution + weakening proof
      without resolutions on $x_{k+1}$.
  \item\label{c2:2} If $c_s(v)\neq c_D(v)$ for some $v\in D$, then $\exists w\in\dcl_D(M)\setminus \{M\}$ such that $c_D(v) = \res(c_s(w), c_s(v))$ on $x_{k + 1}$, and this resolution is half-ordered. In fact, $w$ is a parent of some $w_i$ in $D$, $i\in[s]$.
\end{enumerate}
\end{claim}
\begin{proof}
\hfill

\noindent\ref{c2:1}. No new resolution on the variable $x_{k+1}$ has been introduced, while all old ones are in $A_1\cup\dots\cup A_s$ and thus have been eliminated. The conclusion follows from this observation together with Claim \ref{cl:p1,1}\eqref{c1:4}.
\medskip

\noindent\ref{c2:2}. Suppose $c(v)$ was changed in $D_{i-1}\to D_i$ (and hence stayed unchanged afterwords) so that in particular $v\in A_i$. Set $w\coloneqq w'$, where $w'$ is the parent of $w_i$ from above. Note that $c_s(w)=c_D(w)$ since the latter contains the literal (say) $x_{k+1}$. Then we readily have $c_D(v)=c_{D_{i-1}}(v)= \res(B\lor x_{k+1}, c_D(v)\lor \overline{x_{k+1}})$ by \eqref{eq:subclause}, and it is half-ordered since $B$ is $k$-small.
\end{proof}

Now to get $\Pi_{k+1}$, we try to reconnect $D_s$ with $U$ along $L_k$ and then clear out weakenings. The problem with this approach is the added appearances of $x^a_{k + 1}$ in $c_s(v)$ for $v \in L_k$, as in Claim \ref{cl:p1,1}\eqref{c1:2}. We introduce new nodes to deal with them. Namely, for $v\in L_i$, if $c_s(v)\neq c_{D}(v)$, apply Claim \ref{cl:p1}\eqref{c2:2} to create a new node $\tilde{v}$ to add to $D_s$ with the clause $\res(c_s(w), c_s(v))$. Denote by $\widetilde\Pi_{k + 1}$ the result of connecting $D_s$ and $U$ along $L_k$ and this new set of vertices. Since neither $D_s$ nor $U$ contain resolutions on $x_{k+1}$, $\widetilde \Pi_{k+1}$ is a half-ordered refutation (with weakenings) that is {\em ordered up to $k+1$}. Let $\Pi_{k+1}$ be obtained by contracting all weakening rules. It will still be half-ordered and ordered up to $k+1$, and it only remains to analyze its size (note that \textit{a priori} it can be doubled at every step, which is unacceptable).

Since
\begin{equation} \label{eq:increment}
|\Pi_{k+1}|\leq |\Pi_k| + |L_k|,
\end{equation}
we only have to control $|L_k|$. For that we will keep track of the invariant $|\dcl_{\Pi_k}(L_k)|$; more precisely, we claim that

\begin{equation} \label{eq:induction}
|\dcl_{\Pi_{k+1}}(L_{k+1})| \leq |\dcl_{\Pi_k}(L_k)|.
\end{equation}
\noindent Let us prove this by constructing an injection from $\dcl_{\Pi_{k+1}}(L_{k+1})$ to $\dcl_{\Pi_k}(L_k)$; we will utilize the previous notation.

First note that the resolution + weakening refutation $\widetilde\Pi_{k+1}$
and its weakening-free contraction $\Pi_{k+1}$ can be related as follows. For
every node $v\in V(\Pi_{k+1})$ there exists a node $v^* \in
V(\widetilde\Pi_{k+1})$ with $c_{\widetilde\Pi_{k+1}}(v^*) \supseteq
c_{\Pi_{k+1}}(v)$ which is {\em minimal} among those contracting to $v$. If
$v$ is an axiom node of $\Pi_{k+1}$ then so is $v^*$ in $\tilde\Pi_{k+1}$.
Otherwise, if $u$ and $w$ are the two parents of $v$, and $u',w'$ are the
corresponding parents of $v^*$ ($v^*$ may not be obtained by weakening due to the minimality
assumption), then $c_{\widetilde\Pi_{k+1}}(u^*)$ is a subclause of
$c_{\widetilde\Pi_{k+1}}(u')$ and $c_{\widetilde\Pi_{k+1}}(w^*)$ is a
subclause of $c_{\widetilde\Pi_{k+1}}(w')$. We claim that $(v\mapsto
v^*)\mid_{\dcl_{\Pi_{k+1}}(L_{k+1})}$ (which is injective by definition) is
the desired injection. We have to check that its image is contained in $\dcl_{\Pi_k}(L_k)$.

Fix $v\in\dcl_{\Pi_{k+1}}(L_{k+1})$. Then {\em either $v$ is an axiom or both its parents are not $(k+1)$-small} (by \eqref{eq:$L_i$}). By the above mentioned facts about the contraction $\widetilde\Pi_{k+1}\to \Pi_{k+1}$, this property is inherited by $v^*$. In particular, $v^* \not\in\{\tilde w\mid w\in L_k\}$ as all nodes in this set have at least one $(k+1)$-small parent due to half-orderedness. Finally, since the corresponding clauses in $D$ and $D_s$ differ only in the variable $x_{k+1}$, $v^*$ cannot be in $U$, for the same reason (recall that all axioms are in $D$). Hence $v^* \in
V(D_s)=V(D)=\dcl_{\Pi_{k}}(L_{k})$.

Having thus proved \eqref{eq:induction}, we conclude by the obvious induction that $|L_k|\leq |\dcl_{\Pi_{k}}(L_{k})|\leq |\dcl_{\Pi_{0}}(L_{0})|\leq |\Pi|$. Then \eqref{eq:increment} implies $|\Pi_{n-1}|\leq n|\Pi|$, as desired.
\end{proof}

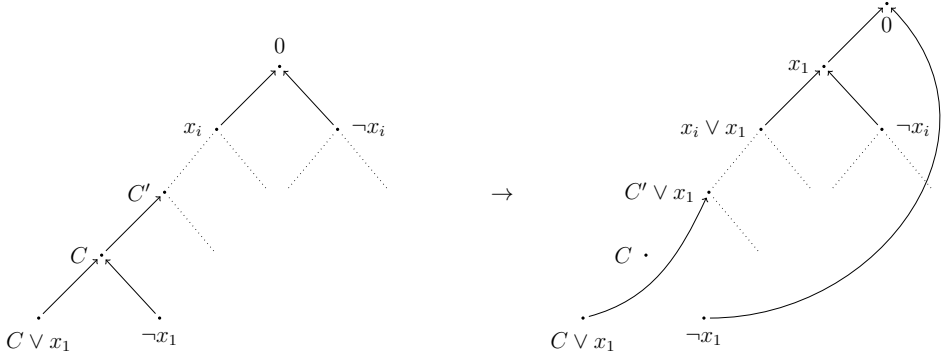
\begin{figure}
\centering
\input{p1-image}
\caption{A toy example of the transformation from $\Pi_0$ to $\Pi_1$}
\end{figure}

\subsection{Half-Ordered $=_p \cdcl(\pid, \decisionl)$ } \label{sec:equiv1}

In this section, we prove the following theorem.

\begin{theorem}\label{thm:cdcl=p1}
The systems $\cdcl(\pid, \decisionl)$ and $\pi$-half-ordered resolution are
$p$-equivalent.
\end{theorem}

One direction is almost trivial.

\begin{proposition}\label{prop:cdcl=p1,1}
$\cdcl(\pid, \decisionl)$ $p$-simulates $\pi$-half-ordered resolution.
\end{proposition}

\begin{proof}
As usual, assume $\pi = \text{id}$. Suppose $C \vee D= \res(C \vee x_i,\ D
\vee \overline{x_i})$ is any half-ordered resolution, and without loss of
generality assume $C$ is $i$-small. It is enough to present a partial run
from $(\tau, \Lambda)$ to $(\tau \cup \{C \vee D\}, \Lambda)$ of length at
most $n + 1$, where $\tau$ is any clause set containing $C \vee x_i$ and $D
\vee \overline{x_i}$.

Let $x_j$ be the largest variable in $C$ (thus $j<i$). Consider a trail of the
form
\[
    t = [x_1 \stackrel d = a_1, \dots,
         x_j \stackrel d = a_j,
         x_i \stackrel u = 1,
         x_{j + 1} \stackrel d = a_{j + 1}, \dots,
         x_{i-1} \stackrel d = a_{i-1}, x_{i+1} \stackrel d = a_{i+1}, \dots,
         x_n \stackrel d = a_n]
\]
such that $(C \vee D) |_t = 0$. By definition, $t[l]
\in \text{AllowedActions}((\tau, t[\leq l-1]))$ for all $l \not = j+1$. But
since
$C$ is $i$-small, $(C \vee x_i) |_{t[\leq j]} = x_i$ and thus $x_i \stackrel
u = 1 \in \text{AllowedActions}((\tau, t[\leq j]))$ as well. Therefore,
\[
    (\tau, \Lambda) \stackrel{t[1]}{\Longrightarrow}
    (\tau, t[\leq 1]) \stackrel{t[2]}{\Longrightarrow}
    (\tau, t[\leq 2]) \dots \stackrel{t[n]}{\Longrightarrow}
    (\tau, t)
\]
is a partial run from $(\tau, \Lambda)$ to $(\tau, t)$.
It now suffices to show $(\tau \cup \{C \vee D\}, \Lambda) \in L((\tau, t))$.
This follows by verifying Definition \ref{df:action} directly: $(D \vee
\overline{x_i}) |_t = 0$ so $D \vee \overline{x_i} \in \C_{n + 1}((\tau, t))$.
For $j' > j + 1$, the assignment $t[j']$ is a decision, so $D \vee \overline{x_i}
\in \C_{j + 2}((\tau, t))$.  Since $(C \vee x_i) |_{t[\leq j]}= x_i$, $C \vee D
= \res(C \vee x_i, D \vee \overline{x_i}) \in \C_{j + 1}((\tau, t))$.  Finally,
for $j' \leq j$, $t[j']$ is a decision, so $C \vee D \in \C_1(\tau, t)$ and
$(\tau \cup \{C \vee D\}, \Lambda) \in \AllowedActions((\tau, t))$.
\end{proof}

The other direction of Theorem \ref{thm:cdcl=p1} is less obvious. We begin with
some additional notation.

Previous works describe standard learning schemes like $\decisionl$ with
respect to so-called \textit{trivial resolution} on a set of particular clauses
(e.g., in \cite{2011-Pipatsrisawat-Darwiche, 2004-Beame-Kautz-Sabharwal}). We
can recast this notion in our model by the following lemma. Let

\[
    D \circ^{x} C \df
        \begin{cases}
            \res(D, C) & \text{if $C$ and $D$ are resolvable on $x$} \\
            D & \text{otherwise (``null case'')}
        \end{cases},
\]
and extend it by left associativity, i.e., $C_0\circ^{x_{i_1}} C_1
\circ^{x_{i_2}}\dots\circ^{x_{i_k}} C_k \df (\dots(C_0\circ^{x_{i_1}}
C_1)\circ^{x_{i_2}}\dots)\circ^{x_{i_k}} C_k$.  We note that if $x_{i_j}$
appears maximally in $C_{j}$ (according to $\pi$) for each $j \in [k]$, then
all the resolutions are $\pi$-half-ordered.

\begin{lemma}\label{lm:learned_deriv}
Assume that a clause $D$ is learned from the state $S=(\C,t = [y_1 \stackrel
{*_1} = a_1,\dots, y_r \stackrel {*_r} = a_r])$ as described in Definition
\ref{df:action}, i.e.\ $D\in\C_j(S)$ for some $j \in [r + 1]$ . Then, there
are clauses $C_1,\dots, C_{k + 1} \in \mathbb{C}$ and indices $j \leq i_1 <
\dots < i_k\le r$ such that
\begin{enumerate}
    \item $C_{k + 1}|_t = 0$,

    \item $C_\nu|_{t[\leq i_\nu - 1]} = y_{i_\nu}^{a_{i_\nu}}$ for $\nu \in [k]$, and

    \item $D = C_{k+1} \circ^{y_{i_k}} C_k \dots \circ^{y_{i_1}} C_{1}$,
        where all operators are not null.
\end{enumerate}

Moreover, let $\ell\in [r]$, $*_\ell=u$ and assume that there exists $\nu\in [k]$ such that $i_\nu>\ell$ and $C_\nu$ contains the literal $\overline{y_{\ell}^{a_{\ell}}}$. Then $\ell\in\{i_1,\ldots,i_k\}$.
\end{lemma}

\begin{proof}
By tracing Definition \ref{df:action} using reverse induction on $j$. If $j =
r + 1$, the clause $D$ itself satisfies the above properties. For $j\leq r$,
$D \in \mathbb{C}_{j + 1}(S)$ or there are clauses $D' \in \mathbb{C}_{j +
1}(S)$ and $C \in \mathbb{C}$ such that $C|_{t[\leq j - 1]} = y_j^{a_j}$ and
$D=\res(D',C)$ on $y_{j}$. The first case follows by the inductive hypothesis.
In the second case, enlarge the index list by including $j$ and the clause
list by including $C$.

The last claim follows from the observation that the literal $\overline{y_{\ell}^{a_{\ell}}}$ will propagate all the way down to $D'\in\mathcal C_{\ell+1}(S)$. Hence, at the $\ell$th stage the second case above must take place.
\end{proof}

In other words, $C_{k + 1}$ is a conflict clause and the other $C_\nu$'s are
clauses in $\C$ chosen to do resolutions while backtracking in a learning step.
These clauses are not necessarily unique, but we fix a choice arbitrarily. For
convenience, we let $i_{k + 1} \df r + 1$ and $
I\df\{i_1,...,i_{k+1}\}$.

Proposition \ref{prop:cdcl=p1,main} completes the proof of Theorem
\ref{thm:cdcl=p1}, which together with Theorem \ref{thm:p1} finishes the proof
of Theorem \ref{thm:last-l}.

\begin{proposition}\label{prop:cdcl=p1,main} $\pi$-half-ordered resolution
$p$-simulates $\cdcl(\pid, \decisionl)$.
\end{proposition}

\begin{proof}
Fix a successful run in $\cdcl(\pid, \decisionl)$. Since the clause set only
changes after a learning step, it suffices to show that for each learning
step $S=(\C, t) \stackrel{(D, t^*)}{\Longrightarrow} (\C \cup \{D\}, t^*)$,
there is a short half-ordered resolution proof of $D$ from $\C$. Suppose $t =
[y_1 \stackrel {*_1} = a_1, \dots, y_r \stackrel {*_r} = a_r]$ and assume
$\pi=\text{id}$, as usual. Fix the clauses $C_\nu$ for $\nu \in [k + 1]$ and
the set $I$ as in Lemma \ref{lm:learned_deriv} with $j\coloneqq 1$ (it is crucial that $D\in
\mathbb C_1(S)$ due to the presence of the amendment $\decisionl$); recall that

\begin{equation} \label{eq:resolution_sequence}
D = C_{k+1} \circ^{y_{i_k}} C_k \dots \circ^{y_{i_1}} C_{1}.
\end{equation}

The sequence of resolutions \eqref{eq:resolution_sequence} is not all
half-ordered only if some $y_{i_\nu}^{a_{i_\nu}}$ is not the largest in
$C_\nu$ (which may happen since the assignments in $t$ need not necessarily
respect the order $\pi$). Thus our goal is to replace in this sequence, this time going from the right to the left, each
clause $C_\nu$ for $\nu \in [k]$ by a clause $C_\nu'$ in which $y_{i_\nu}$
appears maximally.

First, let $C'_{1} = C_{1}$. For $\nu\in[2,k+1]$, let

\begin{equation}\label{eq:c_prime}
C'_{\nu} \df C_\nu \circ^{y_{i_{\nu - 1}}} C_{\nu - 1}' \dots \circ^{y_{i_1}}
C_{1}'
\end{equation}
where this time some operators may be null.

It is immediate from \eqref{eq:c_prime} that
\begin{equation}\label{eq:c_prime1}
\text{$y^{a_{\nu}}_{i_\nu}\in C'_\nu$ for all $\nu\in[k]$}
\end{equation}
and
\begin{equation}\label{eq:c_prime2}
\text{$C'_\nu\subseteq\mathop{\bigcup}_{\mu=1}^\nu C_{\mu}$ for all $\nu\in[k+1]$ (by induction on $\nu$).}
\end{equation}

\begin{lemma}\label{lm:rest}
For all $\nu \in [k+1]$ and $\mu \leq \nu$,
\[
     (C_\nu \circ^{y_{i_{\nu - 1}}} C_{\nu - 1}' \dots \circ^{y_{i_\mu}} C_{\mu}')
|_{t[\leq i_\nu - 1]} = y_{i_\nu}^{a_{i_\nu}}\quad\text{where $y_{i_{k+1}}^{a_{i_{k+1}}} \coloneqq 0$}.
\]
In particular, if $\mu<\nu$ the operator $\circ^{y_{i_\mu}}$ in \eqref{eq:c_prime} can be null only if
$\overline{y_{i_{\mu}}^{a_\mu}} \not \in C_\nu \circ^{y_{i_{\nu - 1}}} C_{\nu - 1}' \circ^{y_{i_{\nu - 2}}}
\dots \circ^{y_{i_{\mu + 1}}} C_{\mu + 1}'$.
\end{lemma}

\begin{proof}
Use double induction, first on $\nu$ and then on $\mu=\nu\ldots 1$. For $\mu=\nu$, this is Lemma \ref{lm:learned_deriv}(2) (and Lemma \ref{lm:learned_deriv}(1) when $\mu=\nu=k+1$). If $\mu<\nu$ then let $E\df  (C_\nu \circ^{y_{i_{\nu - 1}}} C_{\nu - 1}' \dots \circ^{y_{i_\mu+1}} C_{\mu+1}')$; we have to prove that $(E\circ^{y_{i_\mu}}C_\mu')|_{t[\leq i_\nu - 1]}=y_{i_\nu}^{a_{i_\nu}}$ from $E|_{t[\leq i_\nu - 1]}=y_{i_\nu}^{a_{i_\nu}}$. We can assume w.l.o.g.\ that this operator is not null. Now we only have to remark that $C'_\mu\mid_{t[\leq i_\mu -1]}=y_{i_\mu}^{a_{\mu}}$ by the inductive assumption applied to the
pair $\nu \coloneqq \mu,\ \mu \coloneqq 1$. The ``in particular'' part also follows from
this remark.
\end{proof}

By Lemma \ref{lm:rest}, the variable $y_{i_\mu}$ does not appear in $C_\nu \circ^{y_{i_{\nu - 1}}} C_{\nu - 1}' \dots \circ^{y_{i_\mu}} C_{\mu}'$ ($\mu<\nu$). Also, it does not appear in $C_{\mu-1},\ldots,C_1$ (by Lemma \ref{lm:learned_deriv}(2)) and hence also in $C'_{\mu-1},\ldots, C_1'$ (by \eqref{eq:c_prime2}). Therefore it does not appear in $C_\nu'$ and we arrive at the following strengthening of \eqref{eq:c_prime2}:
\begin{equation}\label{eq:c_prime_supset}
    \forall \nu\in[k+1],\ C'_\nu\subseteq (\mathop\bigcup\limits_{\mu=1}^{\nu}
C_\mu)\backslash(\mathop\bigcup\limits_{\mu=1}^{\nu-1}\{y_{i_{\mu}}, \overline{
y_{i_{\mu}}}\}).
\end{equation}
By the last part of Lemma \ref{lm:rest}, \eqref{eq:c_prime_supset} means any variable different from $y_{i_\nu}$ in $C_\nu'$ is marked as $d$ in $t_{i_\nu-1}$. This clearly implies that $y_{i_\nu}$ is maximal in $C'_\nu$. Therefore, for all $\nu \in [k + 1]$ the sequence
$C_{\nu}\circ^{y_{i_{\nu-1}}}C'_{\nu-1}\dots\circ^{y_{i_1}}C'_1$ is
half-ordered. Taken together, these sequences yield a half-ordered derivation
of $C_{k + 1}'$ with $O(k^2)$ steps in total.

Finally, by \eqref{eq:c_prime_supset} $C'_{k+1}\subseteq
(\bigcup\limits_{\mu=1}^{k+1}
C_\mu)\backslash(\bigcup\limits_{\mu=1}^{k}\{y_{i_\mu},\overline{y_{i_\mu}}\})$, where the latter clause is contained in $D$ by Lemma \ref{lm:learned_deriv}(3). This suffices for the proposition since the weakening rule is admissible in
$\pi$-half-ordered resolution.
\end{proof}

%% file: p1-image.tex
\resizebox{\textwidth}{!}{
\begin{tikzpicture}
    [point/.style={circle, line width=0.8mm,color=white,draw,inner sep=0.5mm,fill=black}]

    \node[point] (Cx1) [label=below:$C \vee x_1$] {};
    \node[point] (notx1) [label=below:$\neg x_1$, right=2 of Cx1] {};
    \node[point] (C) [label=left:$C$, above right=1 and 1 of Cx1] {};
    \node[point] (C') [label=left:$C'$, above right=1 and 1 of C] {};
    \node[point] (xi) [label=left:$x_i$, above right=1 and 0.8 of C'] {};
    \node[point] (notxi) [label=right:$\neg x_i$, right=2 of xi] {};
    \node[point] (0) [label=above:$0$, above right=1 and 1 of xi] {};

    \node (rpofC') [below right=1 and 0.8 of C'] {};
    \node (rpofxi) [below right=1 and 0.8 of xi] {};
    \node (lpofnotxi) [below left=1 and 0.8 of notxi] {};
    \node (rpofnotxi) [below right=1 and 0.8 of notxi] {};

    \path[->] (Cx1) edge (C);
    \path[->] (notx1) edge (C);
    \path[->] (C) edge (C');
    \path[->] (xi) edge (0);
    \path[->] (notxi) edge (0);
   
    \path[dotted] (rpofC') edge (C'); 
    \path[dotted] (C') edge (xi);
    \path[dotted] (rpofxi) edge (xi);
    \path[dotted] (lpofnotxi) edge (notxi);
    \path[dotted] (rpofnotxi) edge (notxi);

    \node (transform) [right=1.5 of rpofnotxi] {$\rightarrow$};

    \begin{scope}[shift={(10, 0)}]
        \node[point] (Cx1) [label=below:$C \vee x_1$] {};
        \node[point] (notx1) [label=below:$\neg x_1$, right=2 of Cx1] {};
        \node[point] (C) [label=left:$C$, above right=1 and 1 of Cx1] {};
        \node[point] (C') [label=left:$C' \vee x_1$, above right=1 and 1 of C] {};
        \node[point] (xi) [label=left:$x_i \vee x_1$, above right=1 and 0.8 of C'] {};
        \node[point] (notxi) [label=right:$\neg x_i$, right=2 of xi] {};
        \node[point] (0) [label=left:$x_1$, above right=1 and 1 of xi] {};
        \node[point] (new0) [label=below:$0$, above right=1 and 1 of 0] {};

        \node (rpofC') [below right=1 and 0.8 of C'] {};
        \node (rpofxi) [below right=1 and 0.8 of xi] {};
        \node (lpofnotxi) [below left=1 and 0.8 of notxi] {};
        \node (rpofnotxi) [below right=1 and 0.8 of notxi] {};

        \path[->] (Cx1) edge[out=15,in=245] (C');
        \path[->] (xi) edge (0);
        \path[->] (notxi) edge (0);
   
        \path[dotted] (rpofC') edge (C'); 
        \path[dotted] (C') edge (xi);
        \path[dotted] (rpofxi) edge (xi);
        \path[dotted] (lpofnotxi) edge (notxi);
        \path[dotted] (rpofnotxi) edge (notxi);

        \path[->] (notx1) edge[out=0,in=-45, distance=3.2cm] (new0);
        \path[->] (0) edge (new0);
    \end{scope}
\end{tikzpicture}
}

%% file: first_l.tex
In this section we prove Theorem \ref{thm:main}.
We first show that $\cdcl(\pid,\firstl)$ and $\pi$-$P_0$ (see Definition
\ref{df:p0}) are $p$-equivalent and then prove size upper bounds for $\pi$-$P_0$.

\subsection{$\pi$-$P_0=_p\cdcl(\pid,\firstl)$ } \label{sec:equiv2}

\begin{theorem}\label{thm:cdcl=p0}
For any fixed order $\pi$, the systems $\cdcl(\pid)$, $\cdcl(\pid,\firstl)$ and $\pi$-$P_0$ are
$p$-equivalent.
\end{theorem}
\begin{proof}
Let $\Pi$ be a
$\pi$-$P_0$ refutation of a contradictory CNF $\tau$. We simulate $\Pi$ step-by-step in $\cdcl(\pid,\firstl)$ by
directly deriving each clause in $\Pi$. Suppose we have arrived at a state $(\C,\Lambda)$, where
$\C$ contains both premises in the inference
\begin{equation} \label{eq:rule}
  \ternaryinf
    {C \lor x_i^a} {D \lor x_i^{1-a}} {t}
    {C\lor D},
\end{equation}
as well as all preceding clauses, and assume that all variables in $C$ appear before $x_i$ in $t$.
Let $t = [x_{j_1}\stackrel{*_1}{=}a_1,\ldots,
x_{j_r}\stackrel{*_r}{=}a_r,x_{i}\stackrel{*}{=}a,\ldots]$ and (for ease of notation) $t_s\df t[\leq s]$.  To derive $C
\lor D$, we first build the trail $t_r$; note that since $t$ might be derived in $\Pi$ using the
Unit Propagation rule, the sequence $j_1,\ldots,j_r$ need not necessarily be $\pi$-increasing.

We do it simply by performing the corresponding actions in $\cdcl(\pid,\firstl)$ for decisions and
unit propagations. By induction, assume that we have already built $t_{s-1}$. If $*_s=d$ then
$x_{j_s}$ is the smallest variable according to $\pi$ that is not in $t_{s-1}$, so by definition
$x_{j_s} \deq a_s \in D((\C, t_{s-1}))$. In the case of the Unit Propagation rule ($*_s=u$), there
is a clause $E$ in $\Pi$ preceding \eqref{eq:rule} such that $E|_{t_{s-1}} = x_{s-1}^{a_{s-1}}$. Since $E \in \C$ by
assumption, $x_{j_s} \ueq a_{j_s} \in U((\C, t_{s-1}))$.

Next, we build the trail $[t_r,x_i\ueq a]$ (note that it is different from $t_{r+1}$ if $*=d$). It
is possible since $C \lor x_i^a \in \C$ by our assumption. Then we can further extend it by making
decisions in $\pi$-ascending order on the rest of the variables until $D \lor x_i^{1-a}$ becomes a
conflict clause. Denote the resulting state by $S=(\C, t')$.

Since all assignments after $x_i$ in $t'$ are decisions, $D \lor x_i^{1 - a} \in \C_{r + 2}(S)$,
in the notation of Definition \ref{df:action}.
Therefore, $C \lor D\in \C_{r + 1}(S)$, and hence $(C\lor D,\Lambda)$ is in $\textnormal{AllowedActions}(S)$ even in the presence of $\firstl$.
Induction completes the simulation.

The other direction is more direct: $\pi$-$P_0$ $p$-simulates $\cdcl(\pid)$ by design.
Whenever a run arrives at a state $(\mathbb C,t)$, we infer in $\pi$-$P_0$ all clauses
$C\in\mathbb C$ as well as all suffixes of $t$, including $t$ itself. More specifically, for a transition $(\mathbb C,t)\stackrel A{\Longrightarrow} (\mathbb C',t')$, if $A$ is a decision action or a unit propagation action, then we can derive suffixes of $t'$ using the Decision rule and the Unit propagation rule, respectively.
If $A$ is a learning action, then it suffices to make the following simple observation: by construction, for any $\gamma \in [|t|]$, the clauses in $\C_{\gamma}((\mathbb C, t))$ can be derived from clauses in $\C$ and $\C_{\gamma + 1}((\mathbb C, t))$ using the Learning rule with the trail $t$.

It is easy to see that both simulations increase size by at most a multiplicative factor $n$.
\end{proof}

\subsection{$\pi$-$P_0=_p$ Resolution}\label{sec:main}

We start by considering $\pi$-$P_0$ with the weakening rule, for which the proof of $p$-equivalence with resolution is much easier but still demonstrates some ideas
necessary for the full proof of Theorem \ref{thm:p0}.

\subsubsection{$\pi$-$P_0$ + weakening $=_p$ Resolution}\label{sec:weakening}

We begin with a basic observation: if a clause is falsified by a trail, then weakening it may allow us to derive new trails with the Unit propagation rule.
It is, then, natural to expect that introducing the weakening rule may increase the power of $\pi$-$P_0$.
The following lemma formally demonstrates how to leverage this observation.

\begin{lemma}\label{lm:weakening}
Let $\pi$ be any order.
Suppose $C \lor x$, $D \lor \overline x$ are resolvable clauses and $t$ is a trail such that $(C \lor D)|_t \not = 1$ and $x$ is not assigned by $t$.
For any clause $E$ such that $E|_t = 0$, there is a $\pi$-$P_0 + \textnormal{weakening}$ derivation of $C \lor D$ from $C \lor x, D \lor \overline{x}, E$ and $t$ of length at most $2n + 1$.
\end{lemma}

\begin{proof}
Suppose $(C \lor D)|_t = x_{i_1}^{a_1} \lor x_{i_2}^{a_2} \lor \dots \lor x_{i_k}^{a_k}$ (in an arbitrary order).
Define the trail
\[
  t' \df [
    t,
    x_{i_1} \stackrel u = 1 - a_1,
    x_{i_2} \stackrel u = 1 - a_2,
    \dots,
    x_{i_k} \stackrel u = 1 - a_k,
    x \stackrel u = 0].
\]
so that $(C\lor D)|_{t'}=0$.
Let $t'_\gamma$ denote $t'[\leq |t| + \gamma]$ for $\gamma\in[0, k + 1]$ and define $x_{i_{k + 1}} \df x$, $a_{i_{k + 1}} \df 0$.
For $\gamma \in [0, k + 1]$, define
\[
  E_{\gamma} \df E\lor x^{1-a_{\gamma+1}}_{i_{\gamma+1}}
\]
which can be derived from $E$ by weakening.
Since $E_\gamma|_{t_\gamma} = x_{i_{\gamma + 1}}^{1 - a_{\gamma + 1}}$, $t_{\gamma + 1}$ can be derived from $E_\gamma$ and $t_\gamma$ by the Unit Propagation rule.
Finally, $(C \lor D)|_{t'} = 0$ and $x$ is the last variable in $C$ assigned by $t'$, so $C \lor D$ can be derived using the Learning rule with $t'$.
Altogether, we have used $2(k + 1) + 1 \leq 2 n + 1$ steps.
\end{proof}

\begin{theorem}\label{thm:weakening}
For any order $\pi$, $\pi$-$P_0+\textnormal{weakening}$ is $p$-equivalent to resolution.
\end{theorem}
\begin{proof}
We prove the nontrivial direction.
Let $\Pi$ be a resolution refutation of a contradictory CNF $\tau$.
Assume, without loss of generality, that $\pi = \text{id}$ and all variables appear in $\tau$.
Define
\[
  C_i\df x_1\lor ...\lor x_i.
\]
We construct a $\pi$-$P_0 + \text{weakening}$ derivation of $C_i$ by reverse induction on $i$.
For $i = n$, since $\tau$ is contradictory, there must be a clause containing only positive literals, which can be weakened to $C_n$.
Now suppose $C_{i + 1}$ has been derived.
Using the Decision rule, derive the two trails
\[
  t_i^a\df[x_1\stackrel{d}{=}0,...,x_i\stackrel{d}{=}0, x_{i+1}\stackrel{d}{=}a].
\]
for $a \in \{0, 1\}$.

Let $\Pi_i$ denote the following resolution derivation from
$\Pi|_{t_i^1}$: replace each axiom $A$ with $A \lor C_i$ and then let the added literals be inherited naturally throughout.
Note that $\Pi_i$ ends in $C_i$ and does not contain any appearances of $x_{i + 1}$.
We use $\Pi_i$ to construct the desired proof by deriving each clause in order, applying Lemma \ref{lm:weakening} for each resolution step.

The axioms in $\Pi_i$ can be derived either by weakening or by $\pi$-half-ordered resolution with $C_{i + 1}$, depending on whether $\overline{x_{i + 1}}$ appears in the original axiom of $\Pi$.
For each resolution step $\res(C\lor x_j^{a}, D\lor x_j^{1-a})$ in $\Pi_i$, we can apply Lemma \ref{lm:weakening} with $E \coloneqq C_{i + 1}$, $t \coloneqq t_i^0$, and $x \coloneqq x_j$, noting that $x_j$ is greater than $x_{i + 1}$ by design.
As a result, there is a $\pi$-$P_0 + \textnormal{weakening}$ derivation of $C \lor D$ with $O(n)$ steps.

The number of steps to derive $C_i$ by this process is $O(n|\Pi|)$, so the total number of steps to derive each $C_j$ for $j \in [0, n]$ is $O(n^2|\Pi|)$. Finally note that $n \leq |\tau|\leq |\Pi|$ since all variables appear in $\tau$.
\end{proof}

\subsubsection{Proof of Theorem \ref{thm:p0}}\label{sec:final}

It remains to prove that $\pi$-$P_0$ simulates resolution.
This is the interesting direction of Theorem \ref{thm:p0} and follows from Theorem \ref{thm:final} below.

Throughout this section, assume that $\pi = \text{id}$.
We first introduce operators for \textit{lifting} $\pi$-$P_0$ proofs to include appearances of the literal $x_1$ and \textit{deleting variables} from resolution refutations, both of which we use extensively in the proof of Theorem \ref{thm:final}.

The lifting operator is primarily a bookkeeping mechanism for managing auxiliary appearances of the literal $x_1$ in proofs.

\begin{definition}
Let $\psi$ and $\tau$ be CNFs such that $x_1 \not \in \var(\psi)$ and for each $C \in \psi$, $\tau$ contains either $C$ or $C\lor x_1$.
For $C \in \psi$, define $\lift_\tau(C)$ to be the smallest of these two clauses that is in $\tau$.
For a $\pi$-$P_0$ proof $\Pi$ from $\psi$ define $\lift_\tau(\Pi)$ to be the $\pi$-$P_0$ proof resulting from the following operations on $\Pi$.
\begin{itemize}
    \item Add the derivation of $[x_1 \stackrel d = 0]$ by the Decision rule to the beginning of $\Pi$.
    \item Replace each trail $t$ in $\Pi$ with $[x_1 \stackrel d = 0, t]$.
    \item Replace each axiom $A$ appearing in $\Pi$ with $\lift_\tau(A)$ and then let the added appearances of $x_1$ be naturally inherited throughout the clauses of $\Pi$.
\end{itemize}
\end{definition}

Note that we used a similar lifting operation for resolution in the proof of Theorem \ref{thm:weakening} above.
It is straightforward to verify that $\lift_\tau(\Pi)$ is a $\pi$-$P_0$ proof and if $\Pi$ derives $C$ from $\psi$ then $\lift_\tau(\Pi)$ derives $C$ or $C\lor x_1$ from $\tau$.
Note also that this is only possible because $x_1$ is the smallest variable according
to $\pi$ and hence does not interfere with the Learning rule.
In the proof of Theorem \ref{thm:final}, we will want to construct $\Pi$ but will
only be able to derive clauses in $\tau$, so we construct $\lift_\tau(\Pi)$ instead
and then manage the additional appearances of $x_1$.

The second operator, \textit{variable deletion}, is an analog of restriction for sets of variables as opposed to assignments.
Let $S \subseteq V$ be a set of variables.
For a clause $C$, let $\del_S(C)$ denote the result of removing from $C$ all literals whose underlying variables are in $S$.
For a CNF $\tau$, define $\del_S(\tau) \df \{\del_S(C)$ : $C \in \tau\} \setminus \{0\}$.
Here we see the first interesting feature of variable deletion, namely that we ignore clauses that become $0$ after removing variables from $S$.
But, as we show below, if $\tau$ is contradictory and the subset $S$ is proper then $\del_S(\tau)$ is also contradictory.
This is not true in general for $\tau|_\rho \setminus \{0\}$ of course.

Variable deletion for refutations will be given in Definition \ref{df:deletion}.
It is presented as a (linear time) algorithm that operates on the underlying resolution graph as its input, by recursively changing edges and clauses while nodes keep their identity (although some may be deleted).
This is similar to the approach we took in Section \ref{sec:proofcomplexity}.
In order to more easily keep the node structure fixed, the algorithm first produces a proof in the subsystem of resolution $+$ weakening in which all applications of the weakening rule are dummy (that is, are of the form $\displaystyle{\unaryinf CC}$).
We call proofs in this system {\em generalized resolution proofs}.
We further emphasize that variable deletion is defined only on connected refutations, as connectedness is necessary for the output to be a refutation (cf. Claim  \ref{clm:del-subclause}(\ref{item:straightforward})).
Consequently, we ensure in the proof of Theorem \ref{thm:final} that we only apply it to connected refutations.

\begin{definition} \label{df:deletion}
Let $\Pi$ be a \textbf{connected} resolution refutation of $\tau$ and let $S$ be a \textbf{proper} subset of $\var(\Pi)$.
Let $\Gamma$ be the generalized resolution refutation of $\del_S(\tau)$ whose resolution graph is output by the algorithm below.
The resolution refutation $\del_S(\Pi)$ is the result of contracting dummy applications of the weakening rule in $\Gamma$.
\bigskip

\noindent\textbf{Deletion Algorithm}
\medskip

\begin{enumerate}

  \item \label{del:clause-replace}
    For each axiom node $v$, set $c(v) \gets \del_S(c_\Pi(v))$.
    If $c(v)$ becomes $0$ (that is, when $\var(c_\Pi(v))\subseteq S$), delete it.

  \item \label{del:main-step} Processing nodes in topological order, let $v$ be a resolution node and let $v_1,v_2$ be its parents.

    \begin{enumerate}
      \item
        If {\em both} $v_1$ and $v_2$ were previously deleted, delete $v$ as well.

      \item \label{del:problematic_step}
        If only one of them was deleted or none was deleted but $c(v_1), c(v_2)$ are not resolvable, then one of them, say, $c(v_1)$ is a subclause of  $\del_S(c(v))$ (we will see this in Claim \ref{clm:del-subclause}). Set $c(v)\gets c(v_1)$, and replace incoming edges with a dummy weakening edge from $v_1$.

      \item \label{del:good_step}
        If both $v_1$ and $v_2$ survived and $c(v_1), c(v_2)$ are resolvable, set $c(v)\gets \res(c(v_1), c(v_2))$.
\end{enumerate}
\end{enumerate}
\end{definition}

We claim that this algorithm is well-defined (that is, the condition in step \ref{del:problematic_step} is always met) and that the root vertex $v$ is not deleted and $c(v)=0$ (that is, it produces a generalized resolution refutation of $\del_S(\tau)$). Both statements are immediate corollaries of the following claim.

\begin{claim} \label{clm:del-subclause}
\hfill

\begin{enumerate}
\item \label{item:straightforward} A vertex $v$ is deleted if and only if for every axiom node $w\in \text{dcl}_\Pi(v)$ it holds that $\var(c_\Pi(v))\subseteq S$. In particular:
    \begin{itemize}
    \item The root vertex is not deleted (recall that $\Pi$ is connected);

    \item If $v$ is deleted then $\var(c_\Pi(v))\subseteq S$.
    \end{itemize}

\item \label{item:b} For every remaining vertex $v$, $c(v)$ is a subclause of $\del_S(c_\Pi(v))$.

\item \label{item:c} In the situation of step \ref{del:problematic_step}, there indeed exists $v_i$ such that $c(v_i)$ is a subclause of $\del_S(c_\Pi(v))$.
\end{enumerate}
\end{claim}

\begin{proof}
Item \ref{item:straightforward} is straightforward. The two remaining items \ref{item:b}
and \ref{item:c} are proved by induction, simultaneously with the construction.
The only interesting case is step \ref{del:problematic_step} when precisely one
of the two vertices (say, $v_2$) was deleted. But then $\var(c_\Pi(v_2))\subseteq S$ by Claim \ref{clm:del-subclause}(\ref{item:straightforward}) and hence $c_\Pi(v)$ was obtained by resolving on a variable $x_i$ in $S$. Applying Claim \ref{clm:del-subclause}(\ref{item:b}) to the other parent $v_1$, we see that $c(v_1)$ is a subclause of $\del_S(c_\Pi(v_1))$ which in turn is a subclause of $\del_S(c_\Pi(v))$ since $x_i\in S$.
\end{proof}

One key difference between variable deletion and restriction is that $\Pi|_\rho$ may be trivial, in the sense that it is a single empty clause, while $\del_{\var(\rho)}(\Pi)$ is not.
As a simple example, consider the CNF $\{x_1, \overline{x_1}\lor x_2, \overline{x_2}\}$ and the refutation
\begin{prooftree}
  \AxiomC{$x_1$}
  \AxiomC{$\overline{x_1} \lor x_2$}
  \BinaryInfC{$x_2$}
  \AxiomC{$\overline{x_2}$}
  \BinaryInfC{$0$}
\end{prooftree}
If $\rho =\{x_1=0\}$, then $\Pi|_{\rho}$ is trivial, whereas $\del{\{x_1\}}(\Pi)$ is
\begin{prooftree}
  \AxiomC{$x_2$}
  \AxiomC{$\overline{x_2}$}
  \BinaryInfC{$0$}
\end{prooftree}

The final property of $\del_S(\Pi)$ is that its size can be characterized with respect to the relationship between $\Pi$ and $S$.
This allows us to ``slough off" parts of the $\Pi$ that we might have already seen before.

\begin{lemma}\label{lm:size}
Let $\Pi$ be a connected resolution refutation and let $S \subsetneq \var(\Pi)$.
Let $t$ denote the number of resolution steps $\res(C, D)$ in $\Pi$ on variables in $S$.
Then
\[
  |\del_S(\Pi)| \leq |\Pi| -t.
\]
\end{lemma}
\begin{proof}
By Claim \ref{clm:del-subclause}(\ref{item:b}), all remaining resolution steps \ref{del:good_step} are on variables that do not belong to $S$.
\end{proof}

We now have sufficient machinery to prove Theorem \ref{thm:final}. As is
sometimes useful, the simulation we define is more ambitious than
necessary. Rather than outputting a
refutation, it outputs a proof that derives all literals (as unit clauses)
appearing in the input. The motivation for this is twofold. First, unit
clauses make $\piPz$ significantly more powerful because they grant more
control over the trails that can be derived. In particular, if all literals
appearing in a refutation $\Pi$ have been derived, then $\Pi$ can be
simulated in $n|\Pi|$ steps by directly simulating each resolution appearing
in it. Second, in reference to the deletion operator, all clauses of
$\del_S(\tau)$ can be derived using clauses of $\tau$ and unit clauses $x^0$
and $x^1$ for $x \in S$.

Our simulation algorithm is based on the obvious restrict-and-branch method, by which one recurses on $\Pi|_{\{x_i = 0\}}$ and $\Pi|_{\{x_i = 1\}}$, lifts the resulting proofs to have axioms in $\tau$, and then derives $0$ (if it has not been derived already) by resolving the unit clauses $x_i$ and $\overline{x_i}$.
The clear issue with this approach is that we cannot afford to recurse on \textit{both} restricted proofs: there are parts of $\Pi$ that are ``double counted" as a consequence of its DAG structure and the size may blow up.
But recursing on just $\Pi|_{\{x_i = 0\}}$ may ignore relevant parts of $\Pi$, namely those resolutions on variables not even appearing in $\Pi|_{\{x_i = 0\}}$.
This is the purpose of the deletion operator.
The refutation $\del_{\var(\Pi|_{\{x_i = 0\}})}(\Pi)$ is a refutation with resolutions that correspond to resolutions in $\Pi$ but not in $\Pi|_{\{x_i = 0\}}$, so we can recurse on it without worrying about this double counting issue.
This can be iterated so that we eventually see all literals appearing in $\Pi$ without considering a particular resolution more than once.
So an incomplete but instructive outline of our algorithm is this: recurse on $\Pi|_{\{x_i = 0\}}$ and lift the proof to axioms of $\tau$, iterate the deletion operator to derive all literals appearing in $\Pi$ with possible additional appearances of $x_i$, and then simulate $\Pi|_{\{x_i = 1\}}$ directly to derive $\overline{x_i}$ and remove all additional appearances of $x_i$.

Before we finally state and prove Theorem \ref{thm:final}, we present two simple lemmas that are factored out of the proof to simplify its presentation.
The first essentially states that a variable in a connected refutation must play a nontrivial role, which intuitively should be true if we want to derive its corresponding literals.
The second tells us that once we can directly simulate a connected resolution refutation in $\piPz$, we can also directly simulate a proof of all its literals; this is essentially a stronger version of the observation in the previous paragraphs that is more suited to the goal of deriving all literals.

\begin{lemma} \label{lm:restrictedAxiom}
  Let $\Pi$ be a connected resolution refutation of $\tau$ such that $x \in \var(\Pi)$ and let $\Pi'$ be the downward closure of any appearance of $0$ in $\Pi|_{\{x = a\}}$.
  Then there is a clause $C \in \tau$ that contains $x^{1-a}$.
\end{lemma}

\begin{proof}
  Suppose for contradiction that there is no such clause.
  Then all axioms in $\Pi'$ are axioms in $\Pi$ not containing the variable $x$, so in the standard definition of restriction no edges are contracted and $G(\Pi')$ is a downward-closed subgraph of $G(\Pi)$ with identical labels.
  Since $\Pi$ is connected it has a unique appearance of $0$ (otherwise, $0$ would be the premise of some resolution step that is impossible).
  Therefore $\Pi' = \Pi$ which contradicts the fact that $x \in \var(\Pi)$.
\end{proof}

\begin{lemma} \label{lm:allLitsEasy}
For any connected resolution refutation $\Pi$ of $\tau$, there is a resolution proof from $\tau$ of size at most $|\Pi| + 2n^2$ that derives, as unit clauses, all literals of variables in $\var(\Pi)$.
\end{lemma}

\begin{proof}
  It suffices to note that if literals of all variables in $\res(C \lor x_i^0, D \lor x_i^1)$ have been derived as unit clauses, then there is a proof of size at most $2n$ that derives $x_i^0$ and $x_i^1$.
  This process can be repeated on clauses in $\Pi$ in reverse topological order (skipping clauses for which $x_i^0$ and $x_i^1$ have already been derived).
  Connectedness guarantees that every clause appearing in $\Pi$ (and hence every variable) is processed.
\end{proof}

\begin{theorem} \label{thm:final}
There is a polynomial time algorithm that, given a connected resolution refutation $\Pi$ of $\tau$, outputs a $\piPz$ proof of size $O(n^2|\tau||\Pi|)$ that derives,
as unit clauses, all literals of variables in $\var(\Pi)$.
\end{theorem}

\begin{proof}
We present the algorithm $\psim$ recursively.

\bigskip
\noindent \textbf{Simulation Algorithm} ($\psim$)

\begin{enumerate}

\item
If $|\var(\Pi)| = 1$, then for some variable $x_i$, $\Pi$ contains only a resolution of $x_i$ and $\neg x_i$.
In this case, output the axioms $x_i$ and $\overline{x_i}$.

\item \label{itm:firstRecurse}
Assume without loss of generality that all variables appear in $\Pi$.
Define $\Pi^0$ to be the downward closure of some appearance of $0$ in $\Pi|_{\{x_1 = 0\}}$.
Derive $\lift_\tau(\psim(\Pi^0))$ and let $l_{y, a} \in \{y^a, y^a \lor x_1\}$ for $y \in \mathcal \var(\Pi^0)$ denote the lifted unit clauses appearing in it.
Note that $\Pi^0$ might be trivial, in which case $x_1$ is an axiom in $\tau$ and the next step can be skipped.

\item
If $x_1$ appears in any $l_{y, a}$ from the previous step, then derive $x_1 = \res(l_{y, 0}, l_{y, 1})$.
Otherwise, by Lemma \ref{lm:restrictedAxiom}, there is a clause $C \in \tau$ containing the literal $x_1$.
Derive $x_1$ by consecutively resolving $C$ with literals $\overline{\ell_{y,a}}$, for all $\ell_{y,a}$ in $C$.
We note here that these are half-ordered resolutions and hence admissible in $\piPz$, but we refrain from pointing this out in similar cases below.

\item \label{itm:lifted}
Derive the clauses $\{C \circ^{x_1} x_1 : C \in \tau\}$. At this point we have derived a set of clauses $\tau^\ast$ such that for every clause $C$ in
\[
\psi\df \del_{\{x_1\}}(\tau) \cup \bigcup_{y\in\var(\Pi^0)} \{y^0,y^1\},
\]
$\tau^\ast$ contains either $C$ or $C\lor x_1$.

\item \label{itm:secondRecurse} Set $\mathcal S \gets \var(\Pi^0)$. While
    $\mathcal S \cup \{x_1\} \not = V$ perform the following procedure
    constructing a $\pi$-$P_0$ proof from the set of axioms $\psi$. We
    maintain that at the start of each iteration, all unit clauses in
    $\bigcup_{y \in \mathcal S} \{y^0, y^1\}$ have been derived.
    Also, to make clear, the proof constructed in this step is \textbf{not} part of the output, as
    not all the clauses in $\psi$ have actually been derived, but is used in the next step to derive part of the output.

\begin{enumerate}

\item
Construct the clauses of $\del_{\mathcal S \cup \{x_1\}} (\tau)$ by resolving each clause in $\del_{\{x_1\}}(\tau)$ with the unit clauses $x^a$ for $x \in \mathcal S$.
Then build $\del_{\mathcal S \cup \{x_1\}}(\Pi)$ using the deletion algorithm.

\item
Assume without loss of generality that $\del_{\mathcal S \cup \{x_1\}}(\Pi)$ is connected; otherwise, as usual, take the downward closure of any appearance of $0$.
Construct the proof $\psim(\del_{\mathcal S \cup \{x_1\}}(\Pi))$.

\item
Set $\mathcal S \gets \mathcal S \cup \var(\del_{\mathcal S \cup \{x_1\}}(\Pi))$.
\end{enumerate}

\item \label{itm:liftRecurse}
Since $\del_{\mathcal S \cup \{x_1\}}(\Pi)$ is always nontrivial when $\mathcal S \cup \{x_1\} \not = V$ (this follows from the well-definedness of the Deletion operator, Claim \ref{clm:del-subclause}), the previous step terminates. Call the resulting proof $\Upsilon$. It derives from $\psi$ all unit clauses $x_i^a$ for $i \in [2, n]$.
Derive the proof $\lift_{\tau^\ast}(\Upsilon)$, where $\tau^\ast$ is the set of clauses in step \ref{itm:lifted}.
This proof derives (this time from $\tau$) $l_{i, a} \in \{x_i^a, x_i^a \lor x_1\}$ for $i \in [2, n]$.
It remains to derive $\overline{x_1}$.

\item
For that purpose, it is now possible to build any trail (up to annotations) that extends $[x_1 \deq 0]$ by using the Unit Propagation Rule with the lifted unit clauses from the previous step.
Therefore, we can simulate any resolution proof not containing the variable $x_1$ by directly simulating each resolution step.
Do this to the resolution proof extending $\Pi|_{\{x_1 = 1\}}$ that derives all literals appearing in it (Lemma \ref{lm:allLitsEasy}).

\item
By Lemma \ref{lm:restrictedAxiom}, there is a clause $C \in \tau$ containing $\overline{x_1}$ that appears restricted in $\Pi|_{\{x_1 = 1\}}$.
Derive $\overline{x_1}$ by resolving $C$ with all new literals from the previous step, when possible.

\item
Derive all remaining literals by resolving $l_{i, a}$ with $\overline{x_1}$ when necessary.

\end{enumerate}

Let $f(n, m)$ and $s(n, m)$ be upper bounds on the running time of $\psim$ and the size of $\piPz$ proof output by $\psim$, respectively, when $\psim$ is run on a proof containing at most $n$ variables and whose size is at most $m$.
Our primary focus is understanding the contributions of step \ref{itm:firstRecurse} and \ref{itm:secondRecurse} since the algorithm is called recursively in these steps.
Step \ref{itm:firstRecurse} adds at most $s(n - 1, |\Pi^0|)$ to $s(n, |\Pi|)$ and
\[
  f(n - 1, |\Pi^0|) + O(n\cdot s(n - 1, |\Pi|))
\]
to $f(n, |\Pi|)$.

Suppose that step \ref{itm:secondRecurse} iterates $T(\leq n)$ times.
For $i \in [T]$, define $\mathcal S^i$ to be the state of $\mathcal S$ before the $i^{\text{th}}$ iteration and define $\Pi^i$ to be $\del_{\mathcal S^i \cup \{x_1\}}(\Pi)$.
Then steps \ref{itm:secondRecurse}-\ref{itm:liftRecurse} contribute at most $\sum_{i = 1}^T s(n - 1, |\Pi^i|)$ to the size bound and
\[
  \sum_{i = 1}^T f(n - 1, |\Pi^i|) + O(n|\tau|\cdot |\Pi|)
\]
to the running time bound.

The most important fact here is that, by Lemma \ref{lm:size}, $\sum_{i = 0}^T |\Pi^i| \leq |\Pi|$.
This is because the sets $\var(\Pi^i)$ for $i \in [0, T]$ are pairwise disjoint and so the resolutions in each proof $\Pi^i$ correspond to unique resolutions in $\Pi$.
Note the special case of $\Pi^0$, which uses the fact that restrictions, like variable deletion, have the property that all resolutions in the resulting proof correspond to resolutions in $\Pi$ on the same variable.

The auxiliary operations performed throughout the algorithm (e.g., recreating trails by adding assignments to $x_1$ at the start) are clearly $O(n|\tau|\cdot |\Pi|)$ that yields the bounds
\[
  s(n, |\Pi|)
  \leq \sum_{i = 0}^T s(n - 1, |\Pi^i|)
  + O(n|\tau|\cdot|\Pi|)
\]
and
\[
  f(n, |\Pi|)
  \leq \sum_{i = 0}^T f(n - 1, |\Pi^i|)
  + O\left(\sum_{i = 0}^T s(n - 1, |\Pi^i|) \right)
  + O(n^2 |\tau|\cdot|\Pi|).
\]
By induction on $n$, first for $s$ and then $f$, it follows that $s(n, \Pi) = O(n^2|\tau|\cdot |\Pi|)$ and $f(n, |\Pi|) = O(n^3|\tau|\cdot |\Pi|)$.
\end{proof}

%% file: width.tex
%
%
Our last piece of technical work is Theorem \ref{thm:width}, which demonstrates the limitations of bounded width clause learning in the presence of the ordered decision strategy.
Using the connection to $\piPz$ from the previous section, Theorem \ref{thm:width} follows from a general width lower bound for $\piPz$.
Some of the formulas to which this bound applies have constant width refutations and hence, by Theorem \ref{thm:main}, automatically have polynomial size $\piPz$ refutations. Thus this result also shows that there is no size-width relationship for $\piPz$ like the one for resolution proved by Ben-Sasson and Wigderson \cite{1999-Ben-Sasson-Wigderson}.

\newcommand{\varp}[1]{\textnormal{Var}_\pi^{#1}}
Say that a clause $C$ is \textit{almost-$k$-small} if $|\var (C)\setminus
\varp k| \leq 1$, and that a trail  $t=[x_{i_1} \stackrel{*_1}= a_1,\ldots,
x_{i_r} \stackrel{*_r}= a_r]$ is \textit{$k$-trivial} if for $s\df\min(r,k)$,
all assignments in $t[\leq s]$ are decisions on variables in $\varp k$ in
$\pi$-increasing order: $t[\leq s] = [x_{\pi(1)}\stackrel{d}{=}{a_1},\ldots,
x_{\pi(s)}\stackrel d= a_s]$.

\begin{definition} \label{def:robust}
The order $\pi$ is \textit{$k$-robust} for a contradictory CNF $\tau$ if for any restriction $\rho$ such that $|\var(\rho) \setminus \varp k| \leq 1$, the following properties hold:
\begin{itemize}
  \item the formula $\tau|_\rho$ is minimally unsatisfiable, i.e., all strict subsets of $\tau|_\rho$ are satisfiable;
  \item for all $i \in [n]$, if $x_i \in \var(\rho)$ then there is a clause in $\tau$ that appears restricted in $\tau|_\rho$ (i.e., is not satisfied by $\rho$) and contains the variable $x_i$.
\end{itemize}
\end{definition}

For a CNF $\tau_n$, the \textit{$r$-ary parity substitution} of $\tau_n$, denoted by $\tau_n[\oplus_r]$, is the formula in which for all $i \in [n]$, each variable $x_i$ is replaced with $\bigoplus_{j = 1}^r y_{i, j}$ where the variables $y_{i, 1}, y_{i, 2}, \dots, y_{i, r}$ are new and distinct.
As described, $\tau_n[\oplus_r]$ is technically not a CNF, but its encoding as a CNF is straightforward and natural; see \cite{2013-Nordstrom} for full details.
It is also straightforward to check that whenever $\tau_n$ is minimally unsatisfiable and contains all variables $x_1,\ldots,x_n$, the order $\pi$ on the variables of $\tau_n[\oplus_r]$ given by
\begin{align*}
  &\pi(y_{1, 1}) < \pi(y_{2, 1}) < \dots < \pi(y_{n, 1}) < \\
  &\pi(y_{1, 2}) < \pi(y_{2, 2}) < \dots < \pi(y_{n, 2}) < \dots < \\
  &\pi(y_{1, r}) < \pi(y_{2, r}) < \dots < \pi(y_{n, r})
\end{align*}
is $((r - 2)n)$-robust. In fact, this readily follows from the observation that any restriction $\rho$ as in Definition \ref{def:robust} must leave unassigned at least one variable in each group $\{y_{i,1},\ldots,y_{i,r}\}$.

The following theorem shows that robustness implies large width in $\piPz$.

\begin{theorem}\label{thm:width-piPz}
Let $\tau$ be a contradictory CNF formula and let $\pi$ be an $w$-robust order for $\tau$. Then the width of any $\piPz$ refutation of $\tau$ is at least $w$.
\end{theorem}

\begin{proof}
Assume without loss of generality that $\pi=\text{id}$.
Let $\Pi$ be a $\piPz$ refutation of $\tau$ and let $C$ be the first almost-$w$-small clause appearing in $\Pi$. We will actually prove  that $\varp w \subseteq\var(C)$.

First, we claim that all trails that appear before $C$ in $\Pi$ are $(w + 1)$-trivial.
Suppose otherwise and let $t$ be the first trail in $\Pi$ that is not.
Since $\Pi$ contains all prefixes of $t$, and all such prefixes precede $t$, it follows that $t$ is of the form $[t', x_i \ueq a]$, where $t' = [x_1 \deq a_1, x_2 \deq a_2 \dots, x_j \deq a_j]$ and $j < w + 1$.
Suppose that $t$ follows from $t'$ by the Unit Propagation rule with the clause $D$.
This means $D|_{t'}$ is a unit clause, which implies $D$ is almost-$w$-small, contradicting the assumption that $C$ is the first almost-$w$-small clause in $\Pi$.

It then follows that all resolutions (corresponding to applications of the Learning rule) that appear before $C$ are on variables not in $\varp {w + 1}$.
Indeed, suppose that the inference
\[
  \ternaryinf{D \lor x_i^a}{E \lor \overline{x_i^a}}{t}{D \lor E}
\]
appears before $C$ in $\Pi$.
By the claim in the previous paragraph, $t$ is $(w + 1)$-trivial.
Therefore if $x_i \in \varp {w + 1}$, then it is actually assigned in $t[\leq w + 1]$ and so are all variables appearing in $D$.
This implies $D$ is almost-$w$-small, contradicting the assumption that $C$ is the first such clause.

Finally let $\Pi^*$ be the \textit{resolution refutation} corresponding to $\Pi$; that is, the refutation constructed from $\Pi$ by ignoring all trails.
Let $\Gamma$ be the connected subproof of $C$ in $\Pi^*$ on the downward closure of $C$.
By the remark in the previous paragraph, all resolutions in $\Gamma$ are on variables not in $\varp {w + 1}$.
Lastly, let $\rho$ be any restriction with the domain $\varp {w} \cup \var(C)$ that falsifies $C$, so that $\Gamma|_\rho$ is a refutation of $\tau|_\rho$.
By the first property in the definition of robustness, $\tau|_\rho$ is minimal, which implies that all clauses in $\tau|_\rho$ appear as axioms of $\Gamma|_\rho$.
Therefore, there are paths from these clauses (unrestricted) to $C$ in $\Gamma$.
By the second property of robustness, each variable in $\varp w$ appears at least one of these clauses.
Since all resolutions in $\Gamma$ are on variables not in $\varp {w + 1}$, it follows that $\varp w \subseteq \var(C)$.
\end{proof}

Finally, we prove Theorem \ref{thm:width}, which is restated here for convenience.
The proof is a simple variation of the one above (we only have to make sure that the variables in $\varp w$ appear in a {\em learned} clause).

\begin{theorem}
(Theorem \ref{thm:width} restated)
For any fixed order $\pi$ on the variables and every $\epsilon>0$ there exist contradictory CNFs $\tau_n$
with $w(\tau_n\vdash 0)\leq O(1)$ not provable in {\sf CDCL}{\rm (}$\pi$-\sf D, {\sf WIDTH}-$(1-\epsilon)n${\rm )}.
\end{theorem}

\begin{proof}
The formula used here is $\text{Ind}_m[\oplus_r]$ where $\text{Ind}_m$ is the {\em Induction principle}
\[
x_1\land\bigwedge_{i=1}^{m-1} (\overline{x_i}\lor x_{i+1}) \land \overline{x_m},
\]
and $r$ will be chosen as a sufficiently large constant.
The natural resolution refutation of this formula has width $O(r)$.

Fix $\epsilon > 0$.
Let $R$ be a successful run in $\cdcl(\pid, \mathsf{WIDTH}\textnormal{-}(1 - \epsilon)rm)$ on $\text{Ind}_m[\oplus_r]$ and let $\Pi$ be the natural $\piPz$ simulation of this run given by Theorem \ref{thm:cdcl=p0}.
We begin with some observations about $\Pi$ that are easily verified by examining the proof of Theorem \ref{thm:cdcl=p0}.
First, all clauses learned in $R$ are derived exactly in $\Pi$, in the order they appear in $R$.
Second, for any learning step $(C, t')$ in $R$ from the state $(\C, t)$, the proof $\Pi$ contains the connected subproof of $C$ from $\C$ corresponding exactly to the sequence of resolutions used to learn $C$ (Lemma \ref{lm:learned_deriv}).
Furthermore, the trail $t$ appears before this subproof in $\Pi$.

Let $w = (r - 2)m$ and let $D$ be the first almost-$w$-small clause in $\Pi$.
As in the proof of Theorem \ref{thm:width-piPz}, it follows that $\varp w \subseteq \var(D)$ and all trails appearing before $D$ in $\Pi$ are $(w + 1)$-trivial.
If $D$ is not a learned clause, then it appears in the subproof of some learned clause $C$.
Suppose that $C$ follows from the state $(\C, t)$ in $R$.
As is made clear in Lemma \ref{lm:learned_deriv}, all resolutions in the subproof of $C$ are on variables whose assignments are unit propagations in $t$.
Since $t$ appears before $D$, it is $(w + 1)$-trivial, so none of the variables in $\varp w$ are resolved on to derive $C$.
This implies all variables in $\varp w$ are inherited in $C$ from $D$.

The result follows by taking $r > 2 / \epsilon$ so that $(r - 2)m > (1 - \epsilon)rm$.
\end{proof}

%% file: conclusion.tex
This paper continues the line of research aimed at better understanding theoretical limitations of CDCL solvers. We have focused on the impact of decision strategies, and we have considered the simplest version that always requires to choose the first available variable, under a fixed orderings. We have shown that, somewhat surprisingly, the power of this model heavily depends on the learning scheme employed and may vary from ordered resolution to general resolution.

Practically speaking, the fact that $\cdcl(\pid, \alwaysc, \alwaysu, \decisionl)$ is not as powerful as resolution supports the observation that CDCL solvers with the ordered decision strategy are less efficient than those with more powerful decision strategies like VSIDS.
But, although $\decisionl$ is an asserting learning strategy, most solvers use more efficient asserting strategies like \textit{1-UIP}.
What can be proved if $\decisionl$ is replaced with some other amendment modeling a different, possibly more practical asserting learning scheme?
Furthermore, is it possible that $\cdcl(\pid, \alwaysc, \alwaysu)$ does not simulate general resolution?

Just as in \cite{2004-Beame-Kautz-Sabharwal, 2011-Pipatsrisawat-Darwiche, 2011-Atserias-Fichte-Thurley}, our simulations use very frequent restarts. 
Perhaps the most interesting open question in this area is whether it is actually necessary. In the language we have introduced, this amounts to understanding the power of proof systems  {\sf CDCL}({\sf NEVER-R}) and {\sf CDCL}({\sf ALWAYS-C}, {\sf ALWAYS-U}, {\sf NEVER-R}), the latter version being more oriented towards actual CDCL solvers.

We have also proved that our simulations fail quite badly with respect to width (as opposed to size): there are contradictory CNFs $\tau_n$ refutable in constant width but not belonging to $\cdcl(\pid, \mathsf{WIDTH}\textnormal{-}(1-\epsilon)n)$. The ordering $\pi$ in our result, however, essentially  depends on $\tau_n$. Is a uniform version possible? That is, do there exist contradictory CNFs $\tau_n$ refutable in small width that do not belong to (say) $\mathsf{CDCL}(\pid, \mathsf{WIDTH}\textnormal{-}\Omega(n))$ for {\em any} ordering $\pi$? Another interesting question, extracted from \cite{2011-Atserias-Fichte-Thurley}, asks if $\tau_n$ refutable in small width are always in (say) $\cdcl(\alwaysc, \alwaysu, \mathsf{WIDTH}\textnormal{-}\Omega(n))$.

Finally (cf.\ Remark \ref{rem:positional}) our model is geared towards ``positional solvers'', i.e., those that are allowed to carry along only the set of learned clauses and the current trail. This restriction is of little importance in the theoretical, nondeterministic part of the spectrum, but it will make a big difference if we would like to study dynamic decision strategies like VSIDS, further strengthen the amendments {\sf ALWAYS-C} and {\sf ALWAYS-U} by postulating the behavior in the presence of multiple choices, etc. It would be interesting to develop a rigorous mathematical formalism that would include nonpositional behavior as well.